\documentclass[a4paper,11pt]{article}
\usepackage[a4paper]{geometry}
\usepackage{fullpage}
\date{}

\usepackage[
    style=trad-abbrv,
    maxcitenames=2,
    doi=false,
    url=false,
    isbn=false,
    backend=biber]{biblatex}
\usepackage{graphicx}
\usepackage{subcaption}
\usepackage{framed}
\usepackage{amsmath}
\usepackage{amssymb}
\usepackage{amsfonts}
\usepackage{mathrsfs}
\usepackage{array}
\usepackage{enumerate}
\usepackage{amsthm}  
\usepackage{tikz}
\usepackage[linktocpage=true]{hyperref}
\usepackage{listings}
\usepackage{xcolor}
\usepackage{dsfont}
\usepackage{wrapfig}
\usepackage{float}
\usepackage{verbatim}
\usepackage{bm}
\usepackage{placeins}
\theoremstyle{plain}
\newtheorem{thm}{Theorem}[section] 
\newtheorem{lem}[thm]{Lemma} 
\newtheorem{prop}[thm]{Proposition}

\theoremstyle{definition}

\theoremstyle{remark}
\newtheorem{rem}[thm]{Remark} 
\newcommand{\R}{\mathbb{R}}

\newcommand{\Z}{\mathbb{Z}}

\newcommand{\dd}{\,\mathrm{d}}

\newcommand{\bx}{\mathbf{x}}
\newcommand{\bk}{\mathbf{k}}

\newcommand{\bs}{\mathbf{s}}
\newcommand{\cKr}{\mathcal{K}_R}
\newcommand{\cI}{\mathcal{I}}
\newcommand{\by}{\mathbf{y}}

\title{Critical coupling thresholds for tilted Kuramoto-Vicsek models with a confining potential}

\author{Benedetta Bertoli$^1$, Benjamin D. Goddard$^2$, Grigorios A. Pavliotis$^1$}

\begin{document}

\maketitle

{\small \noindent $^1$Department of Mathematics, Imperial College London, London SW7 2AZ, UK.
\\
benedetta.bertoli21@imperial.ac.uk and pavl@ic.ac.uk
}

{\small \noindent $^2$School of Mathematics and Maxwell Institute for Mathematical Sciences, University of Edinburgh, Edinburgh EH9 3FD, UK. b.goddard@ed.ac.uk}

\begin{abstract}
We study a Kuramoto-Vicsek model of self-propelled particles with periodic boundary conditions subject to a constant angular tilt and a confining potential, and its mean-field (Fokker-Planck) behaviour. In the absence of confinement, the uniform density is stationary and we compute the critical coupling for four normalisation variants of the interaction kernel, showing that the leading instability is always spatially homogeneous. When the confining field is present, the uniform state is no longer stationary. We construct the new steady state perturbatively and apply eigenvalue perturbation theory to derive an explicit formula for the critical coupling as a function of the field strength. The threshold increases quadratically with confinement strength, and the tilt enters through the steady-state correction despite having no effect on the threshold in the absence of confinement. We verify the prediction numerically and derive self-consistency equations for stationary states with general multichromatic potentials.
\end{abstract}
 
\section{Introduction}\label{sec:introduction} 
Collective alignment is one of the basic mechanisms underlying the emergence of order in active matter. Examples range from bacterial suspensions and bird flocks to synthetic self-propelled particles, where local interactions between agents lead, at sufficiently low noise or sufficiently strong coupling strength, to the formation of orientationally ordered states~\cite{vicsek2012collective, marchetti2013hydrodynamics, ramaswamy2010mechanics, peruani2012collective, peruani2008mean, bechinger2016active}.
A basic model in this class is the Vicsek model~\cite{vicsek1995novel}. In this model, particles move with constant speed and update their orientation at discrete time steps by aligning with the average direction of neighbouring particles, subject to noise. The connection between the Vicsek alignment dynamics and the Kuramoto model of phase synchronisation has been studied from several perspectives~\cite{chepizhko2010relation, degond2014hydrodynamics}. In the mean-field limit, one obtains a nonlinear, nonlocal Fokker-Planck equation for the one-particle density. In the absence of external forcing, the spatially homogeneous and rotationally uniform density is a stationary solution. The linear stability of this disordered state has been analysed in the kinetic Vicsek framework in~\cite{peruani2008mean, degond2008continuum, degond2015phase}, among others, where in particular the critical coupling and the nature of the leading (spatially homogeneous) instability are established. For the original Vicsek-style particle system, numerical work showed that the onset of collective motion is discontinuous and is accompanied by travelling bands and phase coexistence~\cite{gregoire2004onset,chate2008collective}. When the spatial variable is absent, the model reduces to the noisy Kuramoto equation, whose stability and bifurcation theory fall within the broader framework of McKean–Vlasov equations on the torus. Phase transitions, long-time behaviour, and bifurcation structure for such equations have been studied in detail in~\cite{carrillo2020long}, with the noisy Kuramoto model as a principal example, and further mathematical results can be found in~\cite{dawson1983critical, delgadino2020diffusive}. 

The purpose of this paper is to study how this picture is modified by the presence of two additional angular forces. The first is a constant angular drift, or tilt, of strength $F>0$, modelling for example an imposed torque, a background rotational bias, or intrinsic circling motion. The second is an external orienting field of strength $h\geq 0$, favouring a preferred direction. Such a term may be viewed as an effective description for external forcing or alignment induced by confinement and boundaries.

The effect of these two perturbations is rather different. The tilt $F$ by itself does not change the classical synchronization threshold, since in the absence of confinement it can be removed by passing to a co-rotating frame. The field $h$, on the other hand, breaks rotational invariance. As a result, the uniform density $\rho_0 = \frac{1}{2\pi}$ is no longer stationary when $h\neq 0$, and the stability problem must be formulated around a non-uniform stationary state depending on $h$. This already changes the structure of the problem at a basic level. Moreover, when both $F$ and $h$ are present, their interaction is nontrivial: although the tilt is inessential at $h=0$, it enters the perturbative correction to the stability threshold once the confining field is switched on.

The effect of the external tilt is that detailed balance is no longer satisfied and that non-equilibrium steady states are created. This leads to new physical phenomena and renders the analysis of the problem more complicated. Even in the absence of interactions, i.e. for a single particle moving in a tilted periodic potential, interesting physical and mathematical phenomena can arise, that go far beyond the equilibrium setting. The calculation of the effective velocity of such a nonequilibrium particle, as a function of the tilt and of the temperature, goes back to Stratonovich~\cite{Straton63}, see also~\cite[Sec. 6.6]{pavliotis2014book}. In fact, our derivation of the self-consistency equation in Section~\ref{subsec:self-consistency} is motivated by Stratonovich's calculation. In addition, giant fluctuations can emerge. More precisely, the effective diffusion coefficient, relative to molecular diffusion, becomes infinite in the limit as the tilt converges to the critical forcing~\cite{reimann_al02}, see also~\cite[Sec. 6.6.]{pavliotis2014book}. This surprising phenomenon has been experimentally verified~\cite{Reimann_al2008}. A complete theoretical analysis of this giant acceleration of diffusion and of the corresponding pinning/depinning phase transitions can be found in~\cite{cheng2015long}. One of the motivations for the present work is to understand the combined effect of a nonequilibrium forcing and of collective behaviour, leading to non-equilibrium phase transitions. 

External orienting fields have also been considered in the Kuramoto literature, where terms of the form $-h\sin(\theta)$ arise naturally in models of driven or forced oscillators. A basic difficulty, already present there, is that the relevant reference (stationary) state is no longer uniform. Most of the available analytical results concern the case without tilt. The novelty in the present work is the simultaneous inclusion of tilt and confinement in a Kuramoto-Vicsek setting.

Kuramoto-style models with tilt have been studied before, for example in~\cite{buendia2021hybrid, cheng2015long, ohta2008critical}; the foundational active rotator model of~\cite{shinomoto1986phase} introduced the constant drive as a mechanism for bistability, and more recently~\cite{odor2023synchronization} has studied the resulting synchronization transitions on brain networks.

In this paper we consider a Kuramoto-Vicsek model with tilt and confinement and study the corresponding stability problem in the mean-field Fokker-Planck formulation. Our main interest is in the onset of orientational order and in the way in which it is affected by the external field. The analysis combines the self-consistency formulation of stationary states, a Fourier representation of the linearised operator, and perturbation theory for isolated eigenvalues.
We now summarise the main results and present the organisation of the paper.

\begin{enumerate}[(i)]
\item In Section~\ref{sec:model} we derive the self-consistency equations for stationary states for a general multichromatic interaction potential. We treat both the full spatially dependent problem in the case $v_0=0$ and the spatially homogeneous reduction when $v_0\neq 0$.

\item In Section~\ref{sec:h0} we revisit the case $h=0$ and carry out the linear stability analysis of the disordered state for four natural normalisations of the spatial interaction kernel, keeping the full spatial dependence. In each case we compute the critical coupling explicitly and show that the dominant unstable mode is always the spatially homogeneous Fourier mode $\bk=0$. In particular, this gives a PDE-level justification for the spatially homogeneous reduction used in earlier work. We also recover the fact that the tilt $F$ does not affect the threshold when $h=0$.

\item Section~\ref{sec:h-neq-0} contains the main perturbative result of the paper. For small $h$ we construct the branch of stationary solutions in the form  $\rho_h = \frac{1}{2\pi} + h \rho_1 + O(h^2)$, linearise the dynamics around $\rho_h$, and use eigenvalue perturbation theory to track the critical eigenvalue. The first-order correction vanishes, and the second-order correction can be computed explicitly. This leads to the expansion $\gamma_c(h) = 2\Gamma + \frac{h^2 \Gamma (3F^2 + 8 \Gamma^2)}{F^2(16\Gamma^2 + F^2)} + O(h^4)$, showing that the external field raises the critical coupling quadratically in $h$. The coefficient depends explicitly on both the tilt $F$ and the noise strength $\Gamma$. We compare this formula with direct numerical computations based on a Fourier-Galerkin discretisation of the linearised operator.

\item In Section~\ref{sec:spatial} we return to the full spatially dependent problem for $h\neq 0$. After decomposition in spatial Fourier modes, the self-propulsion term couples neighbouring angular modes and prevents a direct finite-dimensional Kato reduction for $\bk\neq 0$. We explain this obstruction and give a heuristic argument indicating that the mode $\bk=0$ should remain the most unstable one. Under this assumption, the threshold obtained in the homogeneous setting continues to determine the onset of instability in the full problem.
\end{enumerate}

\section{Model and stationary-state formulation}\label{sec:model}
\subsection{Particle system and mean-field limit}\label{subsec:particle}
We consider a system of $N$ self-propelled particles in $[0,1]^2$. The state of particle $i$ at time $t$ is described by its position $\bx_i(t)\in [0,1]^2$ and its orientation $\theta_i(t)\in[0,2\pi)$, evolving according to
\begin{align}
    \dd \bx_i &= v_0 (\cos(\theta_i),\sin(\theta_i)) \dd t, \label{eq:sde-x}\\
    \dd \theta_i &= \left[F-h\sin(\theta_i)-\gamma\sum_{j\in\Omega_i}K(\bx_i,\bx_j)\sin(\theta_i-\theta_j)\right]\dd t+\sqrt{2\Gamma} \dd B_i(t), \label{eq:sde-theta}
\end{align}
for $1\leq i\leq N$. Here $v_0>0$ is the self-propulsion speed, $F>0$ is the constant angular tilt, $h\geq 0$ is the strength of the confining field, $\gamma>0$ is the alignment strength, and $\Gamma>0$ is the angular diffusion coefficient. The processes $(B_i)_{i=1}^N$ are independent standard Brownian motions. The set $\Omega_i$ denotes the neighbours of particle $i$ within interaction radius $R>0$, and the kernel $K$ specifies the spatial interaction weight, while the precise normalisation of the interaction is introduced at the mean-field level.

Under standard mean-field assumptions, the empirical measure $\mu_t^N = \frac{1}{N}\sum_{i=1}^N \delta_{(\bx_i(t),\theta_i(t))}$ converges to a density $\rho(\bx,\theta,t)$ satisfying the nonlinear, nonlocal Fokker-Planck equation
\begin{equation}\label{eq:FP}
    \partial_t \rho = -v_0(\cos(\theta), \sin(\theta))\cdot\nabla_{\bx} \rho
    -\partial_\theta \Big( \big[F - h\sin(\theta) - \gamma I[\rho]\big]\,\rho\Big)
    + \Gamma\,\partial_\theta^2 \rho.
\end{equation}
The nonlocal alignment interaction is written in the form
\begin{equation}
    I[\rho](\bx,\theta,t)
    =
    \frac{1}{\mathcal N[\rho](\bx,\theta,t)}
    \int_{\R^2}\!\!\int_0^{2\pi}
    \cKr(\bx-\bx')\,\sin(\theta-\theta')\,\rho(\bx',\theta',t)
    \,\dd\theta'\dd\bx'
\end{equation}
where $\cKr$ is a spatial interaction kernel with range $R$, and $\mathcal N[\rho]$ denotes a normalisation factor. Different choices of $\mathcal N[\rho]$ correspond to the variants of the alignment rule considered below.
The terms on the right-hand side of \eqref{eq:FP} correspond respectively to transport in physical space with speed $v_0$ in the direction $(\cos(\theta),\sin(\theta))$, angular drift induced by the tilt, the confining field, and the alignment interaction, and angular diffusion with diffusivity $\Gamma$.  In the different types of normalisation that we will consider, an interaction length scale $R$ is introduced. The limit as the interaction length scale $R$ goes to zero is a very interesting one that can lead to discontinuous phase transitions and dynamical metastability~\cite{leimkuhler2025clusterformationdiffusivesystems, gerber2025formationclusterscoarseningweakly}. In this paper, we will fix the interaction length scale $R > 0$.
We refer to, e.g.~\cite{MartzelAslangul2001, peruani2008mean} for a formal derivation of mean field PDEs and to~\cite{chaintron2022propagationI, chaintron2022propagationII} and the references therein for rigorous results on propagation of chaos.

\subsection{Additive versus Non-Additive Interactions}\label{subsec:normalisations}
The interaction term may be normalised in different ways, depending on the modelling assumptions. A comprehensive analysis of how the chosen normalisation affects the critical threshold under a spatial homogeneity assumption on $\rho$ can be found in~\cite{chepizhko2021revisiting}. In this paper, the authors study two different types of normalisation, and they refer to them as additive versus non-additive interaction; see also the analysis presented in~\cite{Barre_2020}.  Since the precise choice affects the form of the linear stability threshold, we record here the four variants that will be considered later in Section~\ref{sec:h0}.

\begin{enumerate}[(i)]
    \item \textbf{Fully normalised} (normalisation in both $\bx$ and $\theta$):
    \begin{equation}\label{eq:I-full-norm}
        I[\rho](\bx,\theta) = \frac{\displaystyle\int_{|\bx-\bx'| \leq R} \int_0^{2\pi} \sin(\theta-\theta')\,\rho(\bx',\theta')\dd\theta'\dd\bx'}{\displaystyle\int_{|\bx-\bx'| \leq R} \int_0^{2\pi} \rho(\bx',\theta')\dd\theta'\dd\bx'}.
    \end{equation}

    \item \textbf{Unnormalised}:
    \begin{equation}\label{eq:I-unnorm}
        I[\rho](\bx,\theta) = \int_{|\bx-\bx'| \leq R} \int_0^{2\pi} \sin(\theta-\theta')\,\rho(\bx',\theta')\dd\theta'\dd\bx'.
    \end{equation}

    \item \textbf{Partial normalisation in $\theta$}:
    \begin{equation}\label{eq:I-theta-norm}
        I[\rho](\bx,\theta) = \frac{\displaystyle\int_{|\bx-\bx'| \leq R} \int_0^{2\pi} \sin(\theta-\theta')\,\rho(\bx',\theta')\dd\theta'\dd\bx'}{\displaystyle\int_0^{2\pi} \rho(\bx,\theta')\dd\theta'}.
    \end{equation}

    \item \textbf{Partial normalisation in $\bx$}:
    \begin{equation}\label{eq:I-x-norm}
        I[\rho](\bx,\theta) = \frac{\displaystyle\int_{|\bx-\bx'| \leq R} \int_0^{2\pi} \sin(\theta-\theta')\,\rho(\bx',\theta')\dd\theta'\dd\bx'}{\displaystyle\int_{|\bx-\bx'| \leq R} \rho(\bx',\theta)\dd\bx'}.
    \end{equation}
\end{enumerate}
The different normalisations in \eqref{eq:I-full-norm}–\eqref{eq:I-x-norm} reflect different modelling assumptions on how the local orientational information is converted into angular drift. In all cases the numerator represents the total alignment signal generated by particles within the interaction region. The denominator determines whether this signal is interpreted as a total torque or as an averaged one. The unnormalised form \eqref{eq:I-unnorm} corresponds to pairwise additive interactions, where each neighbour contributes independently to the turning rate so that the effective alignment strength increases with local density. By contrast, the fully normalised form \eqref{eq:I-full-norm} corresponds to a non-additive response, in which particles align with the local mean orientation and the turning rate remains bounded as the number of neighbours increases; see the discussion in \cite{chepizhko2021revisiting}. The partially normalised forms \eqref{eq:I-theta-norm} and \eqref{eq:I-x-norm} interpolate between these two extremes by normalising only with respect to one variable.

\subsection{Self-consistency formulation for stationary states}\label{subsec:self-consistency}

We now derive self-consistency equations for the stationary solutions of \eqref{eq:FP}. To accommodate a broader class of models, we replace the monochromatic interaction $\sin(\theta - \theta')$ by the derivative of a multichromatic potential $D(\theta) = -\sum_{k=1}^n a_k \cos(k\theta)$, $a_k > 0$, so that the interaction operator takes the form
\begin{equation}\label{eq:multichromatic-I}
I[\rho](\bx,\theta)
=
\frac{1}{\mathcal{N}[\rho](\bx,\theta)}
\sum_{k=1}^n k a_k
\int_{\R^2}\int_0^{2\pi}
\cKr (\bx,\bx')\sin\big(k(\theta-\theta')\big)\,
\rho(\bx',\theta')\,\dd\theta'\dd\bx'.
\end{equation}
The standard monochromatic case treated in the remainder of the paper corresponds to $n=1$, $a_1 = 1$. Self-consistency equations for stationary states are a classical tool in the analysis of mean-field phase oscillator models, developed systematically in the noisy setting via the nonlinear Fokker-Planck equation, for example, in~\cite{bertini2010dynamical}. In the present multichromatic and spatially extended setting, our derivation builds on \cite{bertoli2024stability}, where analogous equations were derived for multichromatic potentials without a spatial variable, and on \cite{bertoli2025phase}, which introduced spatial dependence via a random graph structure in the case without a confining potential.

\subsubsection{Spatially homogeneous reduction ($v_0 \neq 0$)}\label{subsubsec:homog-reduction}

When $v_0 \neq 0$, the transport term is present and it is not straightforward to derive self-consistency equations for the full spatially dependent problem. However, if we restrict attention to spatially homogeneous solutions $\rho(\bx,\theta,t) \equiv \rho(\theta,t)$, the transport term vanishes. Since $\rho$ is independent of $\bx'$, the angular integrals factor out of the spatial convolution, and the interaction field becomes
\[
    I[\rho](\bx,\theta) = w(\bx)\sum_{k=1}^n k a_k \big[\sin(k\theta)\,r_{k,c} - \cos(k\theta)\,r_{k,s}\big],
\]
where
\[
    r_{k,c} := \int_0^{2\pi}\cos(k\theta)\,\rho(\theta)\dd\theta, \qquad r_{k,s} := \int_0^{2\pi}\sin(k\theta)\,\rho(\theta)\dd\theta
\]
are the global Fourier moments, and $w(\bx) := \int_{\R^2} K(\bx,\bx')\dd\bx'$. For spatial homogeneity to be self-consistent (i.e., for a spatially homogeneous $\rho(\theta)$ to produce an $\bx$-independent drift) the function $w$ must be constant: $w(\bx) \equiv \kappa$ for some $\kappa > 0$. Both standard kernels satisfy this condition:
\begin{itemize}
    \item \emph{Unnormalised}: $w(\bx) = \int_{|\bx-\bx'|\leq R}\dd\bx' = \pi R^2$ (on the full plane, or on a periodic domain), so $\kappa = \pi R^2$.
    \item \emph{Fully normalised}: since $\int_0^{2\pi}\rho(\theta)\dd\theta = 1$ for a spatially homogeneous density, the local mass is $n(\bx) = \pi R^2$, giving $w(\bx) = \pi R^2 / n(\bx) = 1$, so $\kappa = 1$.
\end{itemize}
Under this assumption the interaction field simplifies to
\begin{equation*}\label{eq:I-homog}
    I[\rho](\theta) = \kappa\sum_{k=1}^n k a_k \big[\sin(k\theta)\,r_{k,c} - \cos(k\theta)\,r_{k,s}\big],
\end{equation*}
and the angular drift $b(\theta) := F - h\sin(\theta) - \gamma I[\rho](\theta)$ defines the angular potential $U$ via $\partial_\theta U = -b$, giving
\begin{equation*}\label{eq:U-homog}
    U(\theta) = -F\theta - h\cos(\theta) - \gamma\kappa\sum_{k=1}^n a_k\big[r_{k,c}\cos(k\theta) + r_{k,s}\sin(k\theta)\big].
\end{equation*}
Integrating the stationary Fokker-Planck equation once in $\theta$ gives the standard expression for the density of a one-dimensional Fokker-Planck equation with periodic forcing and constant drift:
\begin{equation}\label{eq:self-consist-homog}
    \rho(\theta) = \frac{1}{Z}\,e^{-\Gamma^{-1} U(\theta)} \int_\theta^{\theta + 2\pi} e^{\Gamma^{-1} U(\varphi)}\dd\varphi,
\end{equation}
where $Z$ is the normalising constant ensuring $\int_0^{2\pi}\rho(\theta)\dd\theta = 1$. Together with the definitions of $r_{k,c}$ and $r_{k,s}$ in terms of $\rho$, equation~\eqref{eq:self-consist-homog} constitutes a closed self-consistency system for the spatially homogeneous stationary state. In what follows we set $\kappa = 1$ for ease of notation.

\begin{rem}[Periodicity of the self-consistency integral]\label{rem:periodicity}
When $F \neq 0$, the potential $U$ is quasi-periodic and the integral $\int_\theta^{\theta+2\pi} e^{\Gamma^{-1} U(\varphi)}\dd\varphi$ is not $2\pi$-periodic in $\theta$. To obtain a representation that is both periodic and convenient for numerical evaluation, we substitute $t = \varphi - \theta$ and write
\begin{equation}\label{eq:calI}
    \cI(\theta) := \int_0^{2\pi} \exp\bigg(-\Gamma^{-1}\bigg[Ft + h\cos(t+\theta) + \gamma\kappa\sum_{k=1}^n a_k\big[r_{k,c}\cos(k(t+\theta)) + r_{k,s}\sin(k(t+\theta))\big]\bigg]\bigg)\dd t.
\end{equation}
We can verify directly that $\cI(\theta + 2\pi) = \cI(\theta)$: shifting the integration variable in $\cI(\theta + 2\pi)$ and using the $2\pi$-periodicity of all trigonometric factors absorbs the quasi-periodic part, since the linear term $Ft$ depends only on the integration variable $t$ and not on $\theta$. This representation also has the practical advantage that the integration limits are independent of $\theta$.
\end{rem}

One can prove that \eqref{eq:self-consist-homog} has a unique solution through a standard Banach fixed point argument, for large enough noise strength $\Gamma$.

\begin{lem}\label{lem:uniqueness-homog}
For $\Gamma$ sufficiently large, the map $T:\mathcal{P}\to\mathcal{P}$ defined by
\[
    (T\rho)(\theta) := \frac{1}{Z[\rho]}\,e^{-\Gamma^{-1} U[\rho](\theta)} \int_\theta^{\theta + 2\pi} e^{\Gamma^{-1} U[\rho](\varphi)}\,\mathrm{d}\varphi,
\]
is a contraction on $(\mathcal{P}, \|\cdot\|_{L^1})$, and therefore admits a unique fixed point.
\end{lem}

\begin{proof}
We work throughout on the space $\mathcal{P}$ of $2\pi$-periodic probability densities on $[0,2\pi]$, regarded as a closed subset of $L^1([0,2\pi])$. We first note that $T$ maps $\mathcal{P}$ into itself, which follows from the definition and positivity of $Z[\rho]$. We now establish the contraction. For $\rho,\tilde\rho\in\mathcal{P}$, the two potentials differ only through the Fourier moments $r_{k,c}$ and $r_{k,s}$ , so
\begin{equation}\label{eq:U-Lip}
    \|U[\rho]-U[\tilde\rho]\|_{L^\infty} \leq 2\gamma\sum_{k=1}^n |a_k|\,\|\rho-\tilde\rho\|_{L^1} =: \tilde{A}\,\|\rho-\tilde\rho\|_{L^1},
\end{equation}
using $|\cos(k\theta)|,|\sin(k\theta)|\leq 1$ and $|r_{k,c}[\rho]-r_{k,c}[\tilde\rho]|\leq\|\rho-\tilde\rho\|_{L^1}$. Setting $K[\rho](\theta) := e^{-\Gamma^{-1}U[\rho](\theta)}\int_\theta^{\theta+2\pi}e^{\Gamma^{-1}U[\rho](\varphi)}\,\mathrm{d}\varphi$, the exponential inequality $|e^a - e^b| \leq |a-b|e^{\max(a,b)}$ and the uniform bound $\|U[\rho]\|_{L^\infty}\leq C_0$ for all $\rho\in\mathcal{P}$ together give
\[
    \|K[\rho]-K[\tilde\rho]\|_{L^\infty} \leq C_1\Gamma^{-1}\|U[\rho]-U[\tilde\rho]\|_{L^\infty},
\]
for a constant $C_1 > 0$ depending only on $C_0$. For the normalisation, since $Z[\rho] = \int_0^{2\pi}K[\rho]\,\mathrm{d}\theta$ and $Z[\rho]\geq z_0>0$ uniformly, we have $|Z[\rho]^{-1} - Z[\tilde\rho]^{-1}|\leq z_0^{-2}|Z[\rho]-Z[\tilde\rho]|$. Decomposing $T\rho - T\tilde\rho = Z[\rho]^{-1}(K[\rho]-K[\tilde\rho]) + K[\tilde\rho](Z[\rho]^{-1}-Z[\tilde\rho]^{-1})$ and taking $L^1$-norms, the above estimates give
\[
    \|T\rho - T\tilde\rho\|_{L^1} \leq C_2\,\Gamma^{-1}\|\rho-\tilde\rho\|_{L^1},
\]
for a constant $C_2>0$ independent of $\Gamma$ and of $\rho,\tilde\rho\in\mathcal{P}$. For $\Gamma > C_2$ the map $T$ is therefore a contraction on the complete metric space $(\mathcal{P},\|\cdot\|_{L^1})$, and the Banach fixed-point theorem allows us to conclude uniqueness of the fixed point.
\end{proof}

\subsubsection{Full spatial dependence ($v_0 = 0$)}\label{subsubsec:full-spatial}

When the assumption of spatial homogeneity is dropped, a self-consistency equation for the full spatially dependent stationary state can still be derived, provided the transport term is absent, i.e.\ $v_0 = 0$. In this case, the stationary Fokker-Planck equation reduces to
\begin{equation*}\label{eq:stat-v0}
    0 = -\partial_\theta\Big(\big[F - h\sin(\theta) - \gamma I[\rho](\bx,\theta)\big]\,\rho(\bx,\theta)\Big) + \Gamma\,\partial_\theta^2 \rho(\bx,\theta).
\end{equation*}
We introduce the local Fourier moments
\[
    m_{k,c}(\bx) = \int_0^{2\pi} \cos(k\theta)\,\rho(\bx,\theta)\dd\theta, \qquad m_{k,s}(\bx) = \int_0^{2\pi} \sin(k\theta)\,\rho(\bx,\theta)\dd\theta,
\]
and the spatially convolved moments
\[
    C_k(\bx) := \int_{\R^2} K(\bx,\bx')\, m_{k,c}(\bx')\dd\bx', \qquad S_k(\bx) := \int_{\R^2} K(\bx,\bx')\, m_{k,s}(\bx')\dd\bx'.
\]
Applying the identity $\sin(k(\theta-\theta')) = \sin(k\theta)\cos(k\theta') - \cos(k\theta)\sin(k\theta')$, the interaction field becomes
\begin{equation*}\label{eq:I-fourier}
    I[\rho](\bx,\theta) = \sum_{k=1}^n k a_k \big[\sin(k\theta)\,C_k(\bx) - \cos(k\theta)\,S_k(\bx)\big].
\end{equation*}
The same integration argument as in \S\ref{subsubsec:homog-reduction}, now applied pointwise in $\bx$, gives the self-consistency equation
\begin{equation}\label{eq:self-consist-full}
    \rho(\bx,\theta) = \frac{1}{Z(\bx)}\,e^{-\Gamma^{-1} U(\bx,\theta)} \int_\theta^{\theta+2\pi} e^{\Gamma^{-1} U(\bx,\varphi)}\dd\varphi,
\end{equation}
where the potential $U$ is defined by $\partial_\theta U = -b$ with $b(\bx,\theta) := F - h\sin(\theta) - \gamma I[\rho](\bx,\theta)$, giving
\begin{equation*}\label{eq:potential-U}
    U(\bx,\theta) = -F\theta - h\cos(\theta) - \gamma\sum_{k=1}^n a_k\big[C_k(\bx)\cos(k\theta) + S_k(\bx)\sin(k\theta)\big],
\end{equation*}
and $Z(\bx)$ is the normalising constant ensuring $\int_0^{2\pi}\rho(\bx,\theta)\dd\theta = 1$. Together with the definitions of $C_k$ and $S_k$ in terms of $\rho$, equation~\eqref{eq:self-consist-full} constitutes a closed self-consistency system for the stationary state.

\begin{rem}
The uniqueness result of Lemma~\ref{lem:uniqueness-homog} extends to the self-consistency equation \eqref{eq:self-consist-full} by the same Banach fixed-point argument, working now in the space $L^\infty(\mathbb{R}^2;L^1([0,2\pi]))$ equipped with the norm $\|\rho\|_{L^\infty_\mathbf{x} L^1_\theta} := \sup_{\mathbf{x}}\int_0^{2\pi}|\rho(\mathbf{x},\theta)|\,\mathrm{d}\theta$. The argument proceeds pointwise in $\mathbf{x}$, and the only additional assumption required is that the convolution kernel satisfies
\[
    \kappa := \sup_{\mathbf{x}\in\mathbb{R}^2}\int_{\mathbb{R}^2}|K(\mathbf{x},\mathbf{x}')|\,\mathrm{d}\mathbf{x}' = C < \infty,
\]
under which the Lipschitz estimate \eqref{eq:U-Lip} becomes $\|U[\rho]-U[\tilde\rho]\|_{L^\infty_{\mathbf{x},\theta}} \leq \tilde{A}C \,\|\rho-\tilde\rho\|_{L^\infty_\mathbf{x} L^1_\theta}$, with $\tilde{A}C$ in place of $\tilde{A}$. As all the kernels we consider are bounded functions on $[0,1]^2$, this assumption always holds.
\end{rem}

\section{Reference case: $h=0$}\label{sec:h0}

Having established the self-consistency framework for stationary states, we turn to the linear stability analysis of the disordered state. We begin with the case $h=0$, $F$ arbitrary, in which the uniform density $\rho_0 = \frac{1}{2\pi}$ is a stationary solution of \eqref{eq:FP}, and the classical question is to determine the critical alignment strength $\gamma_c$ at which $\rho_0$ loses stability. We work with the monochromatic interaction ($n=1$, $a_1=1$) throughout this section.

\subsection{Linearisation about the uniform state}\label{subsec:linearisation}

We linearise \eqref{eq:FP} around $\rho_0$ by writing $\rho(\bx,\theta,t) = \rho_0 + \varepsilon\,\phi(\bx,\theta,t)$, $0 < \varepsilon \ll 1$. Since the interaction kernel is odd in $\theta - \theta'$, we have $I[\rho_0] = 0$ in all four normalisation regimes, so that $\rho\, I[\rho] = \rho_0\, I_{\mathrm{lin}}[\phi] + O(\varepsilon^2)$, where $I_{\mathrm{lin}}$ denotes the linearised interaction operator. The exact form of this operator will depend on the choice of normalisation, but the structure of the linearised equation is the same for all cases:
\begin{equation}\label{eq:linearised-h0}
    \partial_t \phi = -v_0(\cos(\theta),\sin(\theta))\cdot\nabla_\bx \phi
    + \Gamma\,\partial_\theta^2 \phi
    + \gamma\rho_0\,\partial_\theta I_{\mathrm{lin}}[\phi]
    - F\,\partial_\theta\phi.
\end{equation}
We expand $\phi$ in spatial and angular Fourier modes:
\[
    \phi(\bx,\theta,t) = \sum_{\bk \in \Z^2}\sum_{m \in \Z}
    \hat\phi_m(\bk,t)\, e^{i\bk\cdot\bx}\, e^{im\theta}.
\]
In this representation, the diffusion and tilt terms are diagonal:  $\Gamma\,\partial_\theta^2 e^{im\theta} = -\Gamma m^2\, e^{im\theta}$, and 
$-F\,\partial_\theta e^{im\theta} = -imF\, e^{im\theta}$. In particular, the contribution of the tilt $F$ is purely imaginary and hence does not influence the real part of the spectrum.

The transport term couples neighbouring angular Fourier modes: using trigonometric identities, we find that
\[
    \hat\phi_m(\bk) \;\longmapsto\;
    -\frac{iv_0}{2}\Big[(k_x + ik_y)\,\hat\phi_{m-1}(\bk)
    + (k_x - ik_y)\,\hat\phi_{m+1}(\bk)\Big],
\]
where $\bk = (k_x, k_y)$. The resulting infinite tridiagonal matrix is skew-Hermitian, and therefore its eigenvalues are purely imaginary. Thus, neither the tilt nor the transport term affects the location of the stability threshold.

It remains to analyse the alignment operator.
The angular integration selects only the first harmonics:
\begin{equation*}\label{eq:angular-selection}
    \int_0^{2\pi} \sin(\theta - \theta')\, e^{im\theta'}\dd\theta'
    = \begin{cases}
        -i\pi\, e^{i\theta}, & m = 1, \\
        \phantom{-}i\pi\, e^{-i\theta}, & m = -1, \\
        0, & \text{otherwise}.
    \end{cases}
\end{equation*}
The spatial convolution over the ball $B_R(\bx)$ produces a scalar factor. Specifically, introducing the substitution $\bs = \bx' - \bx$ and evaluating in polar coordinates using the standard Bessel identity $\frac{\dd}{\dd z}(zJ_1(z)) = zJ_0(z)$, one obtains
\begin{equation*}\label{eq:bessel-factor}
    \frac{1}{\pi R^2}\int_{|\bs| \leq R} e^{i\bk\cdot\bs}\dd\bs
    = \frac{2J_1(|\bk|R)}{|\bk|R} =: S_R(|\bk|),
\end{equation*}
where $J_1$ is the Bessel function of the first kind. The function $S_R$ satisfies $ S_R(|\bk|) \leq 1$, with $S_R(|\bk|) = 1$ if and only if $\bk = 0$.

Combining these observations, we find that for each normalisation choice there exists a real constant $A = A(\gamma)$ such that the linearised interaction acts on Fourier modes according to
\[
    \hat\phi_m(\bk) \;\longmapsto\;
    \begin{cases}
        A\, S_R(|\bk|)\, \hat\phi_m(\bk), & m = \pm 1, \\
        0, & m \neq \pm 1.
    \end{cases}
\]
The real part of the growth rate for mode $(\bk,m)$ is therefore
\begin{equation}\label{eq:growth-rate}
    \mathrm{Re}\,\lambda_m(\bk) =
    \begin{cases}
        A\, S_R(|\bk|) - \Gamma, & m = \pm 1, \\
        -\Gamma m^2, & m \neq \pm 1.
    \end{cases}
\end{equation}
It follows that only the modes $m = \pm 1$ can become unstable, and that among these the largest growth rate is attained at $\bk =0$. The critical value of the coupling is therefore determined by the condition $A(\gamma_c) = \Gamma$.

\subsection{Critical thresholds for the four normalisations}\label{subsec:thresholds}

It remains to compute the prefactor $A$ for each normalisation variant introduced in Section~\ref{subsec:normalisations}. The Fourier analysis carried out above applies in all cases; the only difference lies in the linearisation of the interaction term.

\paragraph{Fully normalised interaction}

Writing $I[\rho] = N[\rho]/D[\rho]$ with $N[\rho] = \int_{|\bx-\bx'|\leq R}\int_0^{2\pi}\sin(\theta-\theta')\,\rho\dd\theta'\dd\bx'$ and $D[\rho] = \int_{|\bx-\bx'|\leq R}\int_0^{2\pi}\rho\dd\theta'\dd\bx'$, a first-order expansion gives
\[
    I[\rho_0 + \varepsilon\phi]
    = \varepsilon\,\frac{N[\phi]}{D[\rho_0]} + O(\varepsilon^2)
    = \frac{\varepsilon}{\pi R^2}\int_{|\bx-\bx'|\leq R}\int_0^{2\pi}
    \sin(\theta-\theta')\,\phi(\bx',\theta')\dd\theta'\dd\bx',
\]
The linearised interaction in Fourier variables is therefore $\widehat{I_{\mathrm{lin}}[\phi]}_m(\bk) = \pi\, S_R(|\bk|)\,\hat\phi_m(\bk)$ for $m = \pm 1$ and zero otherwise, giving the prefactor $A = \gamma\rho_0\,\pi = \gamma/2$ and the critical threshold
\begin{equation*}\label{eq:gc-normalised}
    \gamma_c = 2\Gamma.
\end{equation*}

\paragraph{Unnormalised interaction}

Since there is no denominator to expand, $I_{\mathrm{lin}}[\phi] = I[\phi]$ directly. The spatial convolution now contributes the full factor $\pi R^2$, so $\widehat{I_{\mathrm{lin}}[\phi]}_m(\bk) = \pi^2 R^2\, S_R(|\bk|)\,\hat\phi_m(\bk)$ for $m = \pm 1$. This gives $A = \gamma\pi R^2/2$ and
\begin{equation*}\label{eq:gc-unnormalised}
    \gamma_c = \frac{2\Gamma}{\pi R^2}.
\end{equation*}

\paragraph{Partial normalisation in $\theta$}

In this case we write $I[\rho] = J[\rho]/f[\rho]$ where $J[\rho]$ is the unnormalised numerator and $f[\rho](\bx) = \int_0^{2\pi}\rho(\bx,\theta)\dd\theta$ is the angular mass. Since $J$ is linear in $\rho$, the quotient rule gives $I_{\mathrm{lin}}[\phi] = J[\phi]/f[\rho_0] = J[\phi]$. This is identical to the unnormalised case, giving
\begin{equation*}\label{eq:gc-theta-norm}
    \gamma_c = \frac{2\Gamma}{\pi R^2}.
\end{equation*}

\paragraph{Partial normalisation in $\bx$}

Finally, suppose that $I[\rho] = J[\rho]/g[\rho]$, where $g[\rho](\bx,\theta) = \int_{|\bx-\bx'|\leq R}\rho(\bx',\theta)\dd\bx'$. At the uniform state, $J[\rho_0] = 0$ and $g[\rho_0] = \pi R^2 \rho_0 = R^2/2$. Again using $J[\rho_0] = 0$, the linearisation gives $I_{\mathrm{lin}}[\phi] = J[\phi]/g[\rho_0] = (2/R^2)\,J[\phi]$. In Fourier variables, $\widehat{I_{\mathrm{lin}}[\phi]}_m(\bk) = 2\pi\, S_R(|\bk|)\,\hat\phi_m(\bk)$ for $m = \pm 1$.
The prefactor is $A = \gamma\pi$, giving
\begin{equation*}\label{eq:gc-x-norm}
    \gamma_c = \frac{\Gamma}{\pi}.
\end{equation*}

\medskip

For ease of reference, we collect the results in Table~\ref{tab:thresholds}.

\begin{table}[H]
\centering
\renewcommand{\arraystretch}{1.4}
\begin{tabular}{lc}
\hline
\textbf{Normalisation} & \textbf{Critical $\gamma_c$} \\
\hline
Fully normalised & $2\Gamma$ \\
Unnormalised & $2\Gamma/(\pi R^2)$ \\
Partial in $\theta$ & $2\Gamma/(\pi R^2)$ \\
Partial in $\bx$  & $\Gamma / \pi$ \\
\hline
\end{tabular}
\caption{Critical alignment strength for the four normalisation variants. In each case, the instability is carried by the spatially homogeneous modes $(\bk,m) = (0,\pm 1)$.}
\label{tab:thresholds}
\end{table}

\subsection{Justification of spatial homogeneity}\label{subsec:spatial-homog}

A key consequence of the analysis above is that the leading instability is always spatially homogeneous. Indeed, the growth rate \eqref{eq:growth-rate} for the critical modes $m = \pm 1$ satisfies
\[
    \mathrm{Re}\,\lambda_{\pm 1}(\bk) = A\,S_R(|\bk|) - \Gamma
    < A - \Gamma = \mathrm{Re}\,\lambda_{\pm 1}(\mathbf{0})
    \qquad \text{for all } \bk \neq 0,
\]
since $S_R(|\bk|) < 1$ whenever $\bk \neq 0$. All modes with $m \neq \pm 1$ have $\mathrm{Re}\,\lambda_m = -\Gamma m^2 < 0$ and are therefore always stable. Consequently, the loss of stability of $\rho_0$ is triggered by the spatially homogeneous modes $(\bk,m) = (0,\pm 1)$, and the critical thresholds in Table~\ref{tab:thresholds} coincide with those obtained under a spatial homogeneity assumption.
This provides a justification, at the level of the linearised mean-field equation, for restricting the analysis to spatially homogeneous perturbations when determining the onset of instability.

Finally, we comment on the role of the tilt parameter $F$ in the $h=0$ regime.
The tilt enters the linearised equation \eqref{eq:linearised-h0} through the operator $C\phi = -F\,\partial_\theta\phi$, which acts on the angular Fourier basis as $C\,e^{im\theta} = -imF\,e^{im\theta}$. This contribution is purely imaginary and therefore shifts only the imaginary parts of the eigenvalues, leaving the growth rates \eqref{eq:growth-rate} unchanged.

\begin{prop}\label{prop:F-invariance}
    For $h = 0$, the critical alignment strength $\gamma_c$ is independent of the tilt parameter $F$, for all four normalisation variants.
\end{prop}

The invariance has a simple physical interpretation: since the alignment interaction depends only on phase differences $\theta_i - \theta_j$, the transformation $\theta \mapsto \theta + Ft$ removes the tilt from the particle system \eqref{eq:sde-x}-\eqref{eq:sde-theta} without affecting any relative orientation. The effect of $F$ is therefore only to induce a uniform rotation of the entire population, without affecting the onset of alignment.

To illustrate Proposition \ref{prop:F-invariance}, we performed numerical simulations of the simplified one-dimensional, spatially homogeneous Fokker-Planck PDE
\begin{align*}
    \partial_t \rho = \Gamma \partial_{\theta}^2 \rho + F \partial_{\theta} \rho + \partial_{\theta}(\rho \partial_{\theta}(K * \rho)), \qquad \theta \in [0,2\pi],
\end{align*}
with periodic boundary conditions and kernel $K(\theta) = -\gamma \cos(\theta)$. For a given choice of $F$, we solve the PDE numerically using a Fourier spectral discretisation in $\theta$ with $N=60$ grid points on the interval $[0,2\pi]$. The diffusion coefficient is fixed to $\Gamma = 1$, so that the critical coupling predicted by the linear stability analysis is $\gamma_c = 2$. Simulations are run up to time $t_{\text{max}} = 300$, from a slightly perturbed initial condition $\rho(\theta,0) \propto \sin(\theta) + 2$, which is normalised to integrate to one. To test the independence of the transition from the tilt parameter, we perform a parameter sweep over several values of $\gamma$, taking $\gamma \in \{ 1.5, 2.0, 2.5, 3.5\}$ and repeating the simulations for three different tilt values $F \in \{0, 0.5, 1.0 \}$. Figure \ref{fig:F_invariance_profiles} shows the resulting steady state profiles.
\begin{figure}[t]
\centering

\begin{subfigure}{0.45\textwidth}
    \centering
    \includegraphics[width=\linewidth]{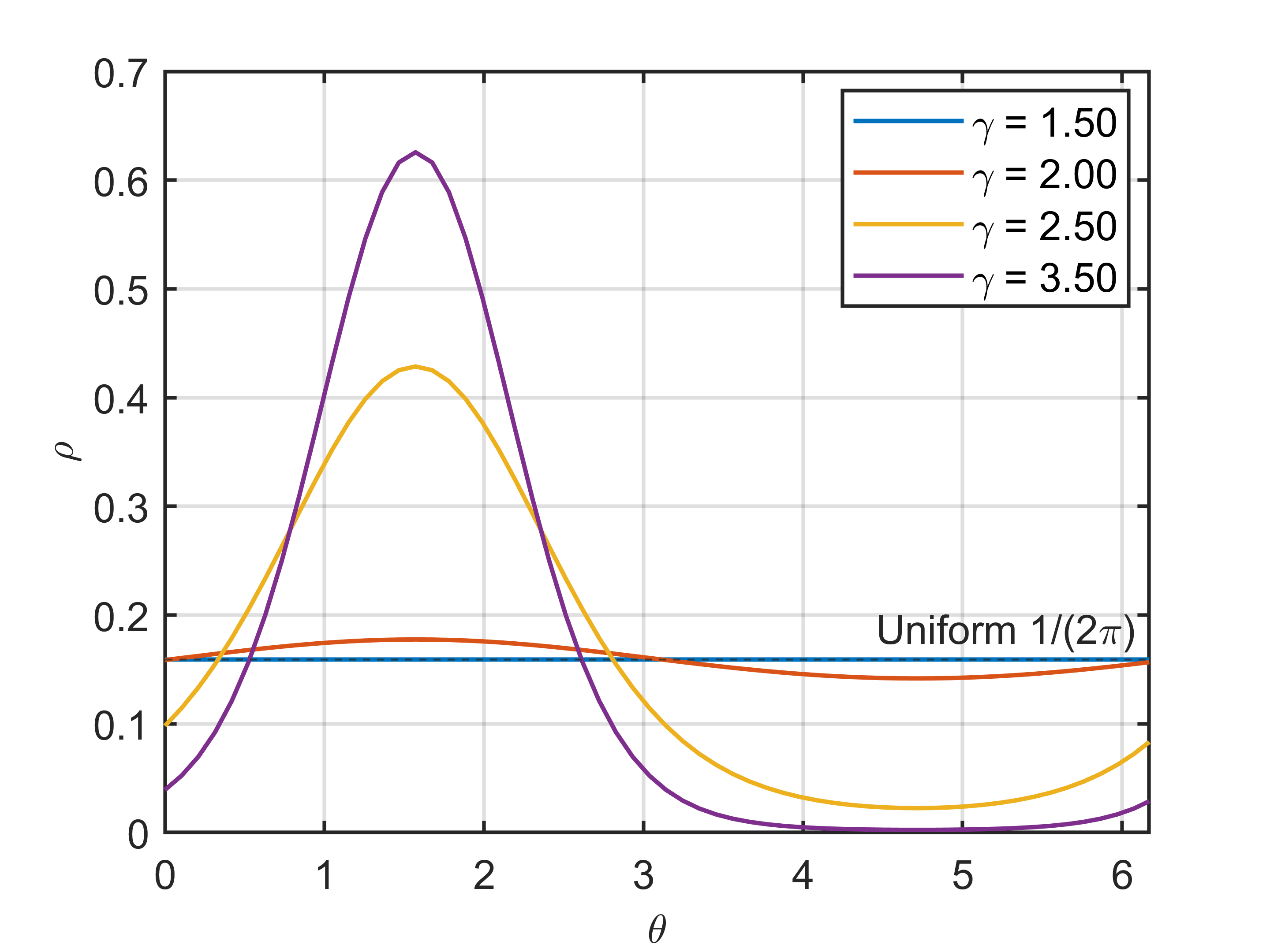}
    \caption{$F = 0$}
\end{subfigure}
\hfill
\begin{subfigure}{0.45\textwidth}
    \centering
    \includegraphics[width=\linewidth]{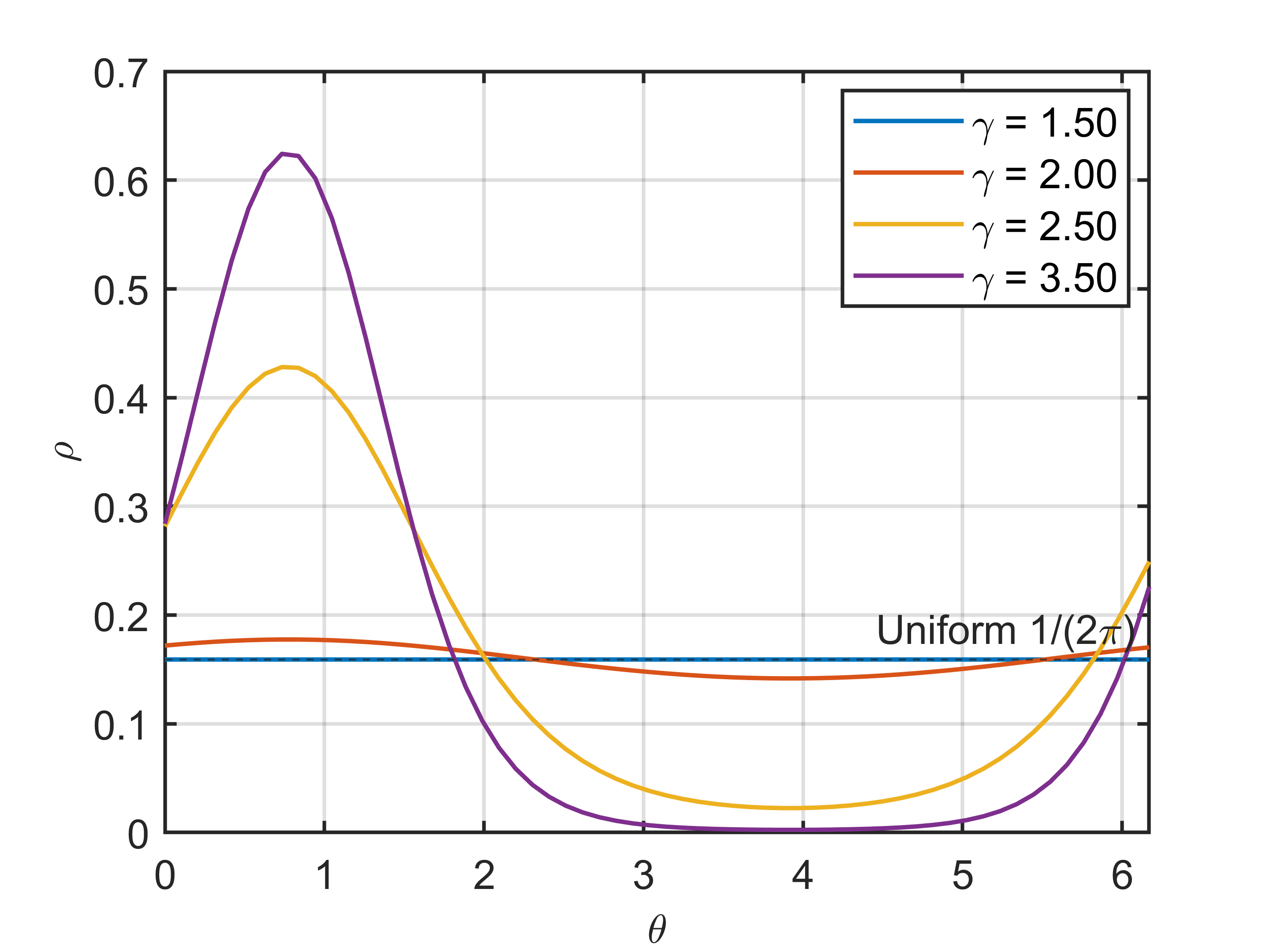}
    \caption{$F = 0.5$}
\end{subfigure}

\vspace{0.5cm}

\begin{subfigure}{0.6\textwidth}
    \centering
    \includegraphics[width=\linewidth]{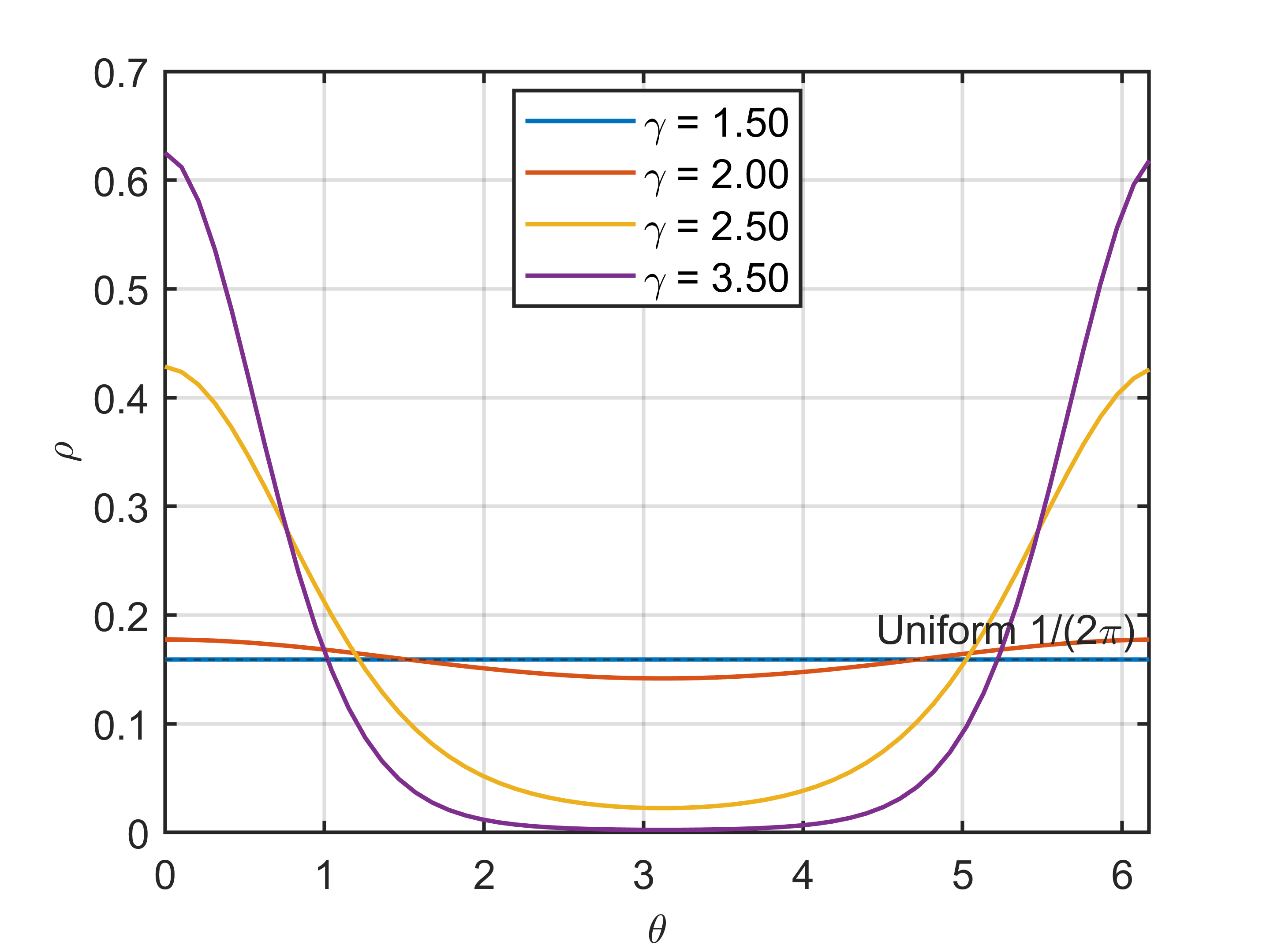}
    \caption{$F = 1$}
\end{subfigure}
\caption{Final density profiles $\rho(\theta,t_{\max})$ for several values of the alignment strength $\gamma$ and tilt parameter $F$. The transition from the uniform state to a clustered profile occurs at the same $\gamma$ for all $F$, consistent with Proposition~\ref{prop:F-invariance}.}

\label{fig:F_invariance_profiles}
\end{figure}

We emphasise, however, that this invariance breaks down when $h \neq 0$. While $F$ alone does not affect the threshold, the combination of $F$ and $h$ does, because the confining field $-h\sin(\theta)$ is not invariant under $\theta \mapsto \theta + Ft$, and the combined effect of $F$ and $h$ is no longer trivial. This is precisely the mechanism that enters the perturbative calculation of the critical coupling in Section~\ref{sec:h-neq-0}.

\subsection{Numerical verification}\label{subsec:numerics-h0}

We verify the analytical thresholds by solving the mean-field equation \eqref{eq:FP} with $h = 0$ numerically in a simplified setting where the spatial domain is one-dimensional, $x \in [0,1]$, with periodic boundary conditions. The density $\rho(x,\theta,t)$ is discretised using a Fourier spectral method with $40$ modes in each variable. Restricting to one spatial dimension does not affect the structure of the stability analysis; the only modification is the replacement of the two-dimensional ball area $\pi R^2$ by the one-dimensional interval length $2R$ in the critical thresholds. The corresponding one-dimensional values are:
\begin{itemize}
    \item Fully normalised: $\gamma_c = 2\Gamma$;
    \item Unnormalised: $\gamma_c = \Gamma/R$;
    \item Partial normalisation in $\theta$: $\gamma_c = \Gamma/R$;
    \item Partial normalisation in $\bx$: $\gamma_c = \Gamma/\pi$.
\end{itemize}

We initialise with a smooth perturbation of the uniform density, $\rho(x,\theta,0) \propto \cos(\theta)\,\sin(2\pi x) + 2$, normalised to unit mass, and evolve until $t = 250$ with parameters $\Gamma = 1$, $R = 0.2$, $v_0 = 0.1$. For each normalisation, we compute steady states at values of $\gamma$ slightly below and above $\gamma_c$. In all four cases, the numerics reproduce the predicted bifurcation structure: for $\gamma < \gamma_c$ the solution relaxes to the uniform state $\frac{1}{2\pi}$, while for $\gamma > \gamma_c$ a non-trivial orientational profile emerges. Representative results are shown in Figures~\ref{fig:unnormalised_pair}-\ref{fig:x-norm}.

\begin{figure}[t]
    \centering
    
    \begin{subfigure}[t]{0.49\textwidth}
        \centering
        \includegraphics[width=\linewidth]{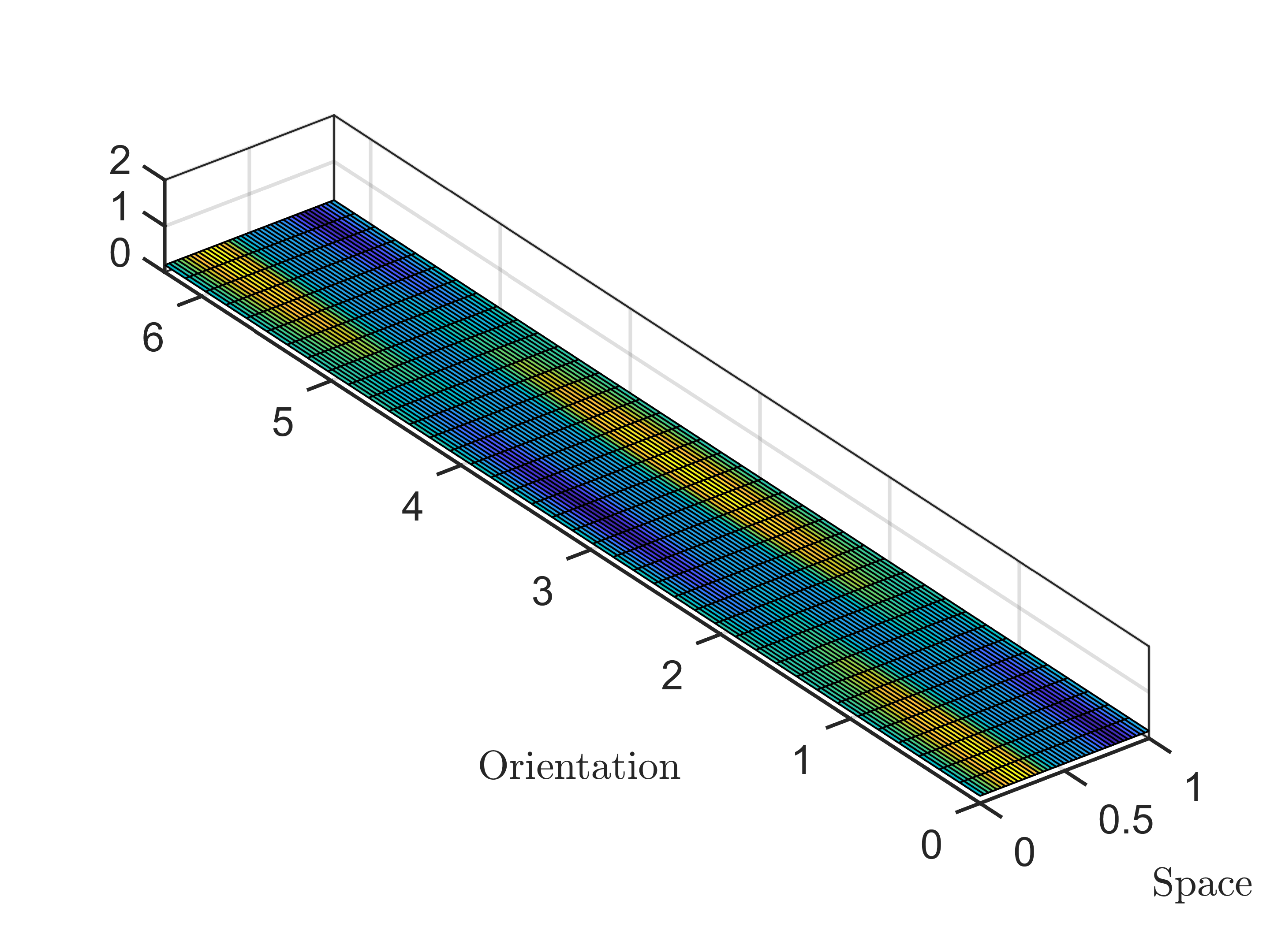}
        \caption{$\gamma = 4.5$}
    \end{subfigure}
    \hfill
    \begin{subfigure}[t]{0.49\textwidth}
        \centering
        \includegraphics[width=\linewidth]{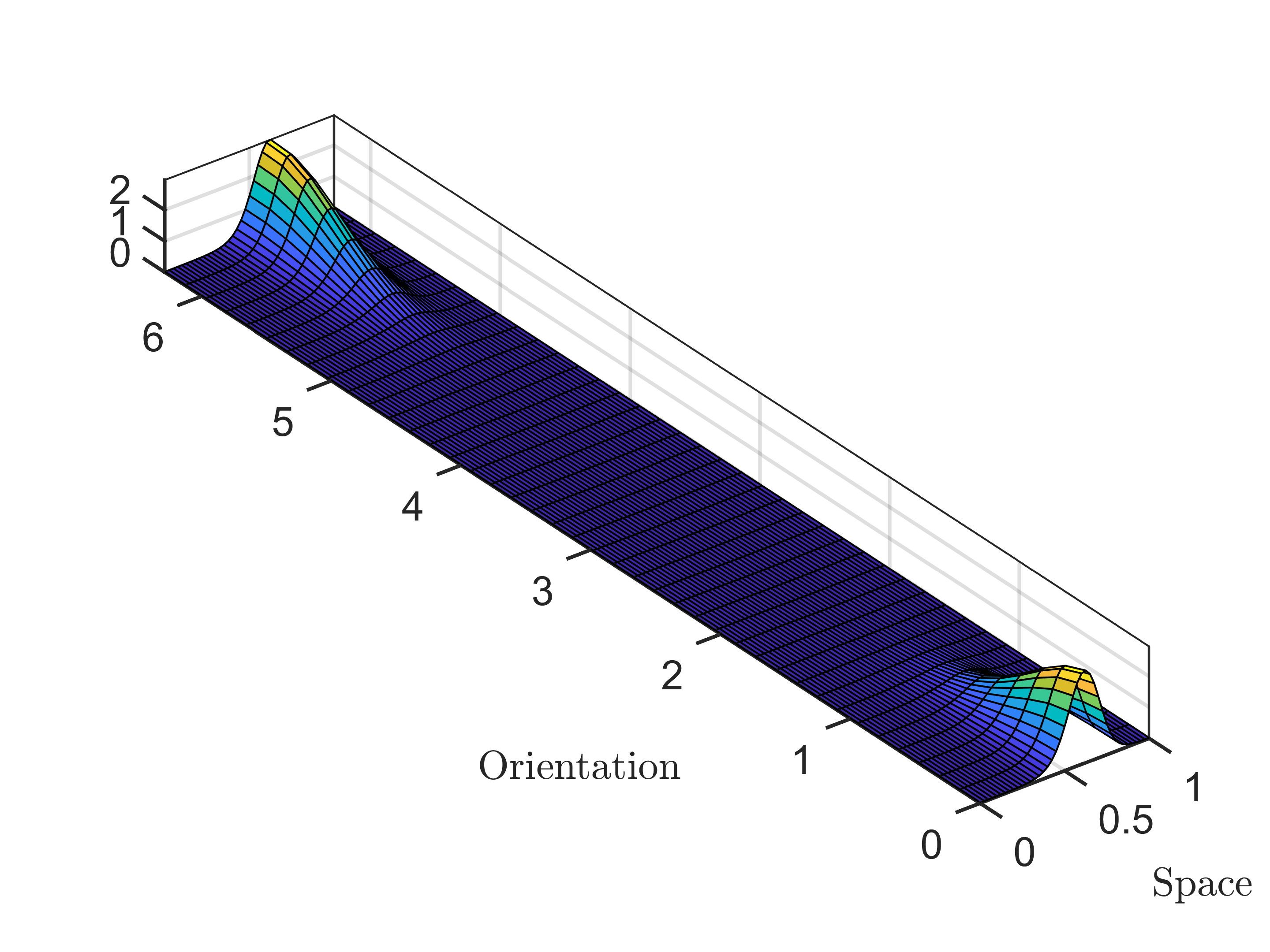}
        \caption{$\gamma = 6$}
    \end{subfigure}
    
    \caption{
    Steady-state density for the unnormalised interaction with $R=0.2$.
    The critical threshold is $\gamma_c=\Gamma/R=5$.
    }
    \label{fig:unnormalised_pair}
\end{figure}
\FloatBarrier

\begin{figure}[t]
  \centering
  \begin{subfigure}[t]{0.49\textwidth}
    \centering
    \includegraphics[width=\linewidth]{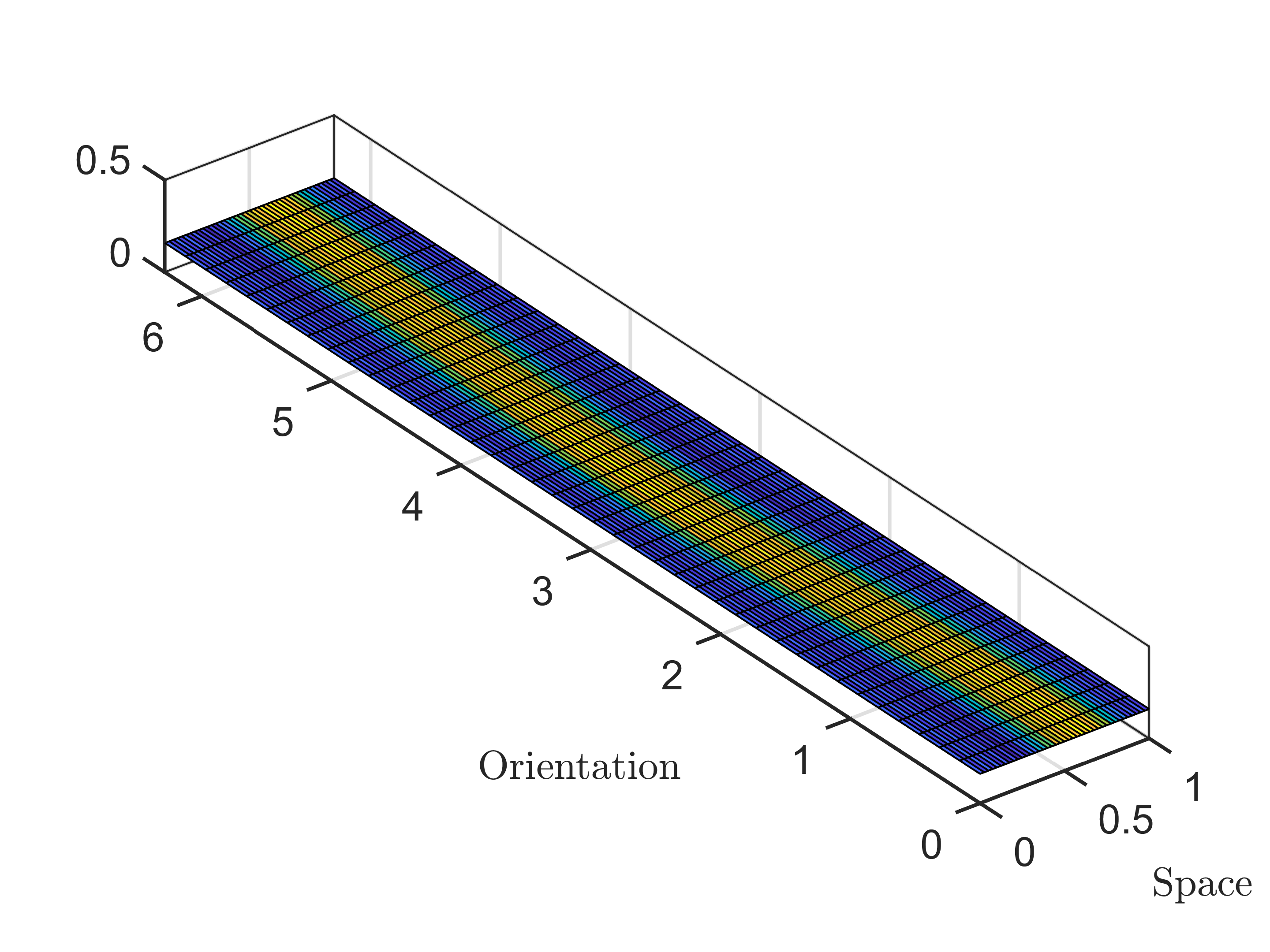}
    \caption{$\gamma = 1$}
  \end{subfigure}
  \hfill
  \begin{subfigure}[t]{0.49\textwidth}
    \centering
    \includegraphics[width=\linewidth]{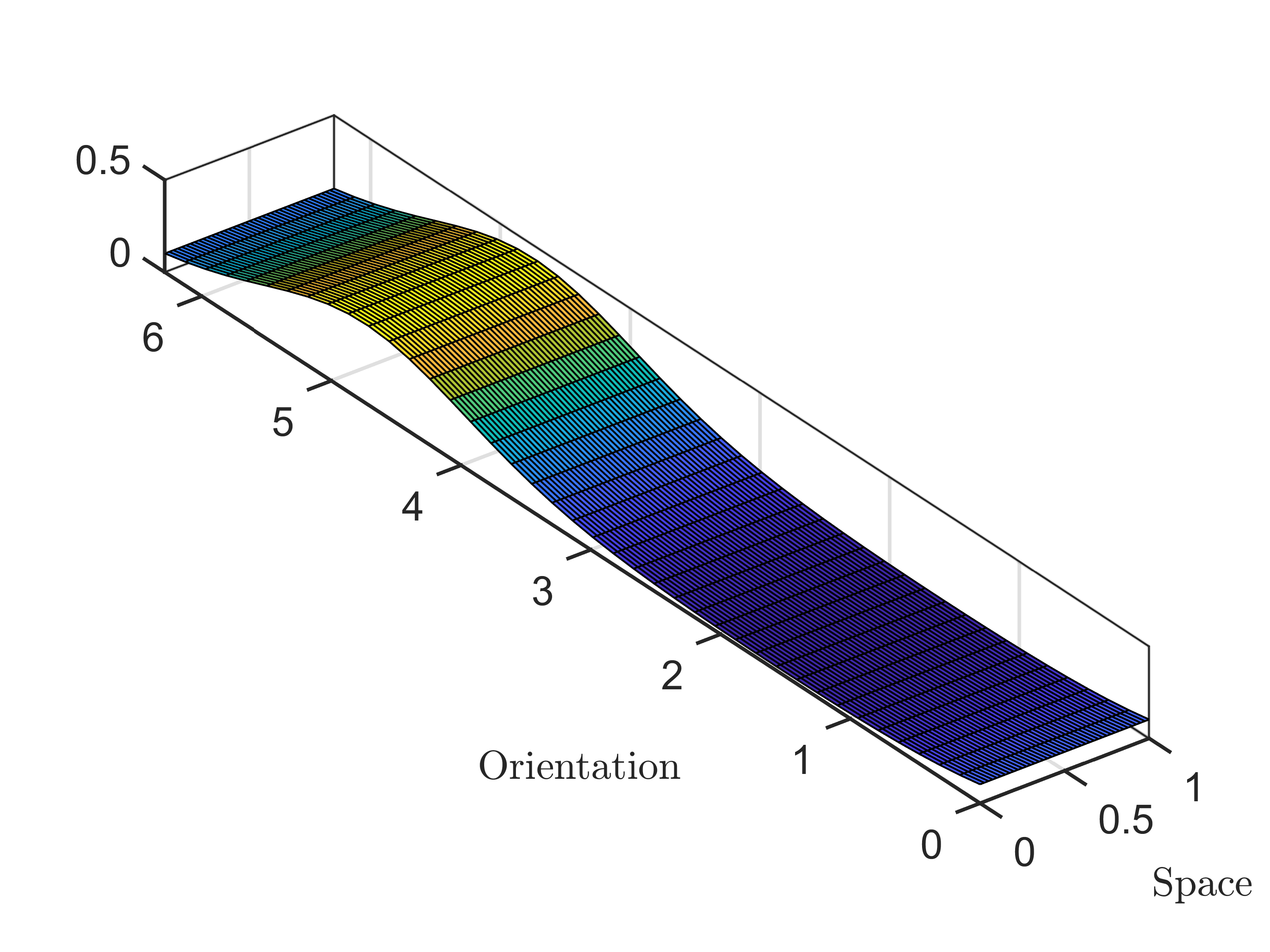}
    \caption{$\gamma = 2.5$}
  \end{subfigure}
  \caption{Steady-state density for the fully normalised interaction.
  The critical threshold is $\gamma_c = 2\Gamma = 2$.
  }
  \label{fig:norm}
\end{figure}

\begin{figure}[t]
  \centering
  \begin{subfigure}[t]{0.49\textwidth}
    \centering
    \includegraphics[width=\linewidth]{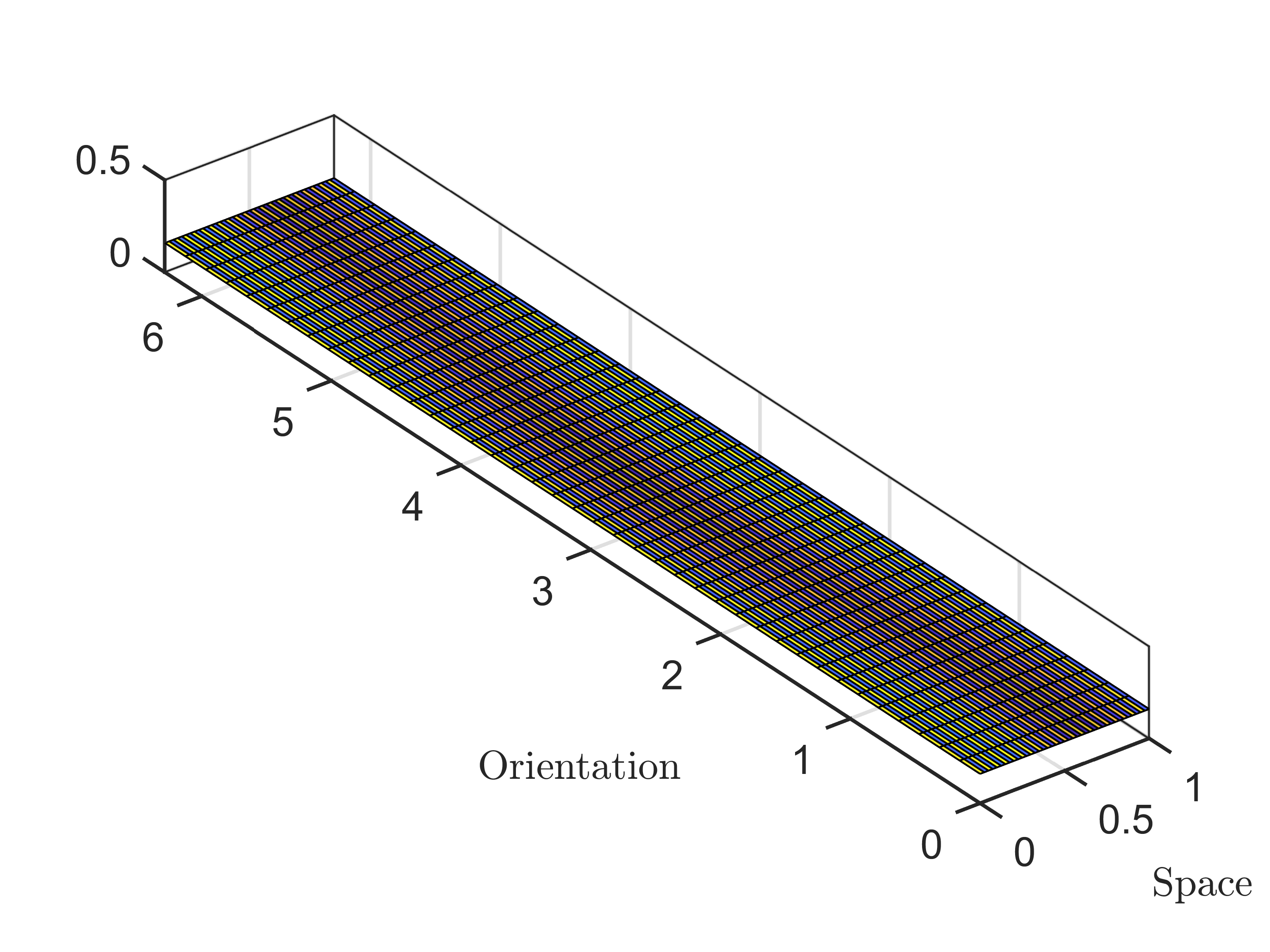}
    \caption{$\gamma = 1$}
  \end{subfigure}
  \hfill
  \begin{subfigure}[t]{0.49\textwidth}
    \centering
    \includegraphics[width=\linewidth]{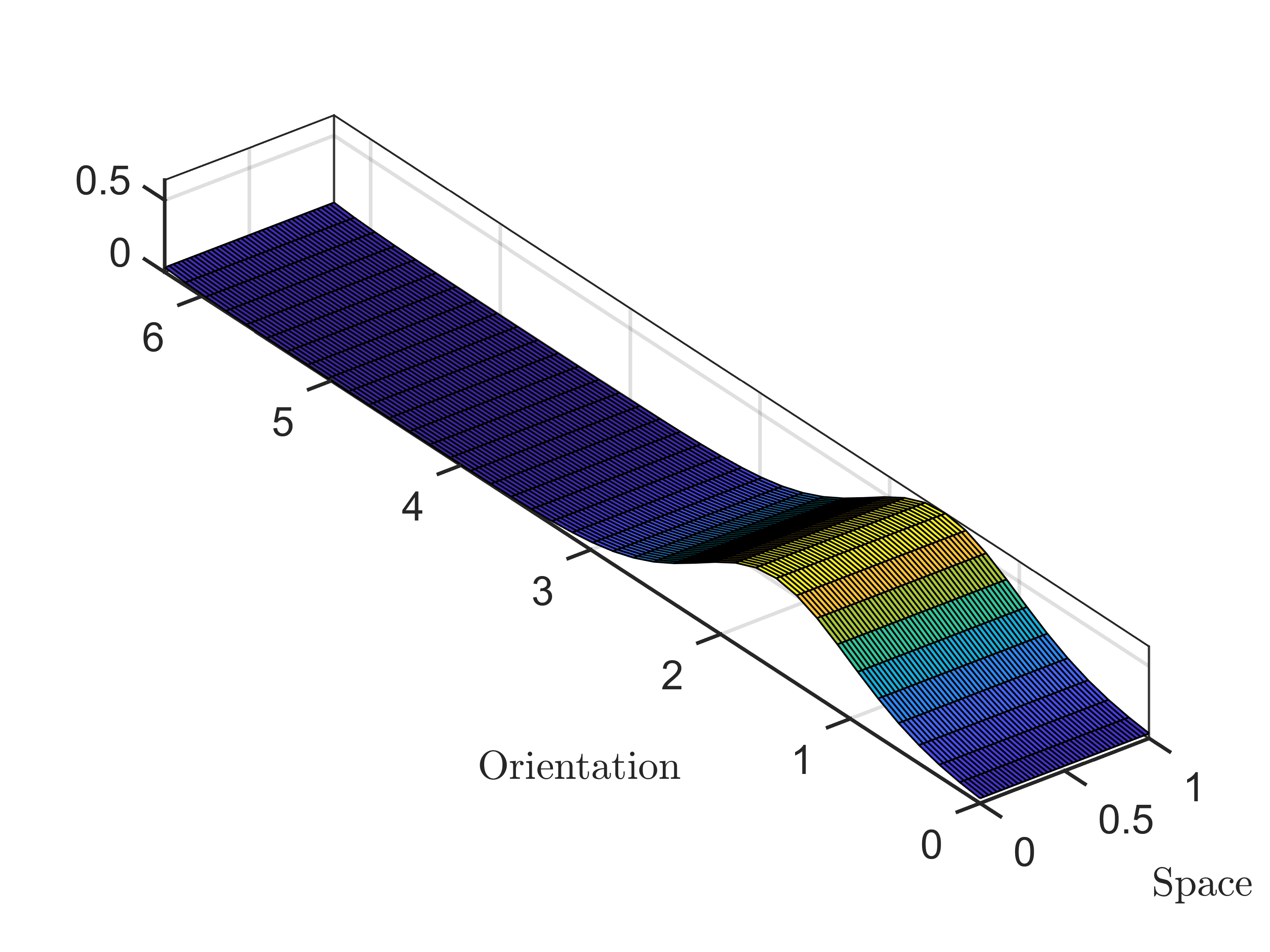}
    \caption{$\gamma = 6$}
  \end{subfigure}
  \caption{Steady-state density for the partially normalised (in $\theta$) interaction with $R = 0.3$.
  The critical threshold is $\gamma_c = \Gamma/R = 3.33$.
  }
  \label{fig:norm_theta}
\end{figure}

\begin{figure}[t]
  \centering
  \begin{subfigure}[t]{0.49\textwidth}
    \centering
    \includegraphics[width=\linewidth]{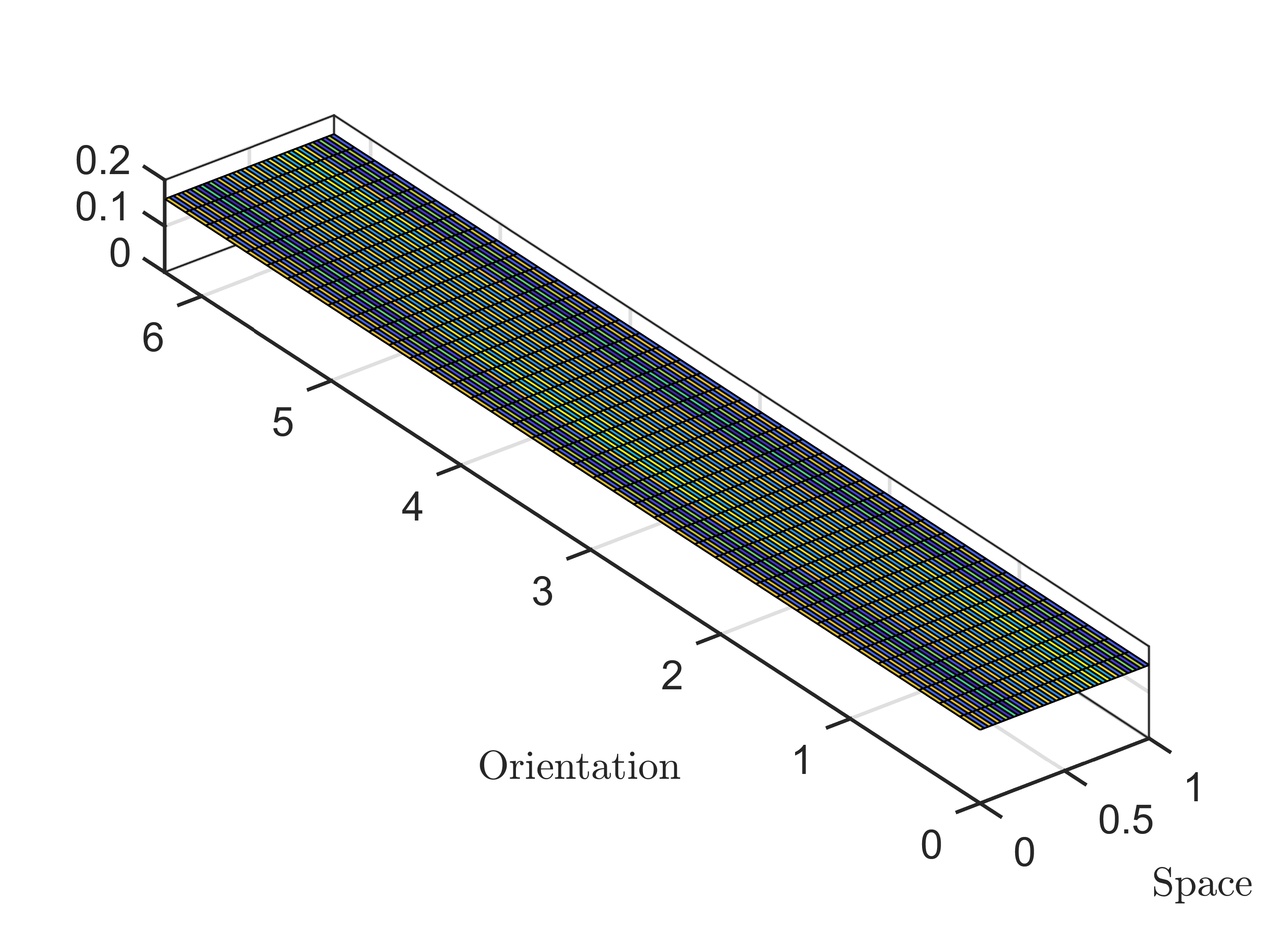}
    \caption{$\gamma = 0.2$}
  \end{subfigure}
  \hfill
  \begin{subfigure}[t]{0.49\textwidth}
    \centering
    \includegraphics[width=\linewidth]{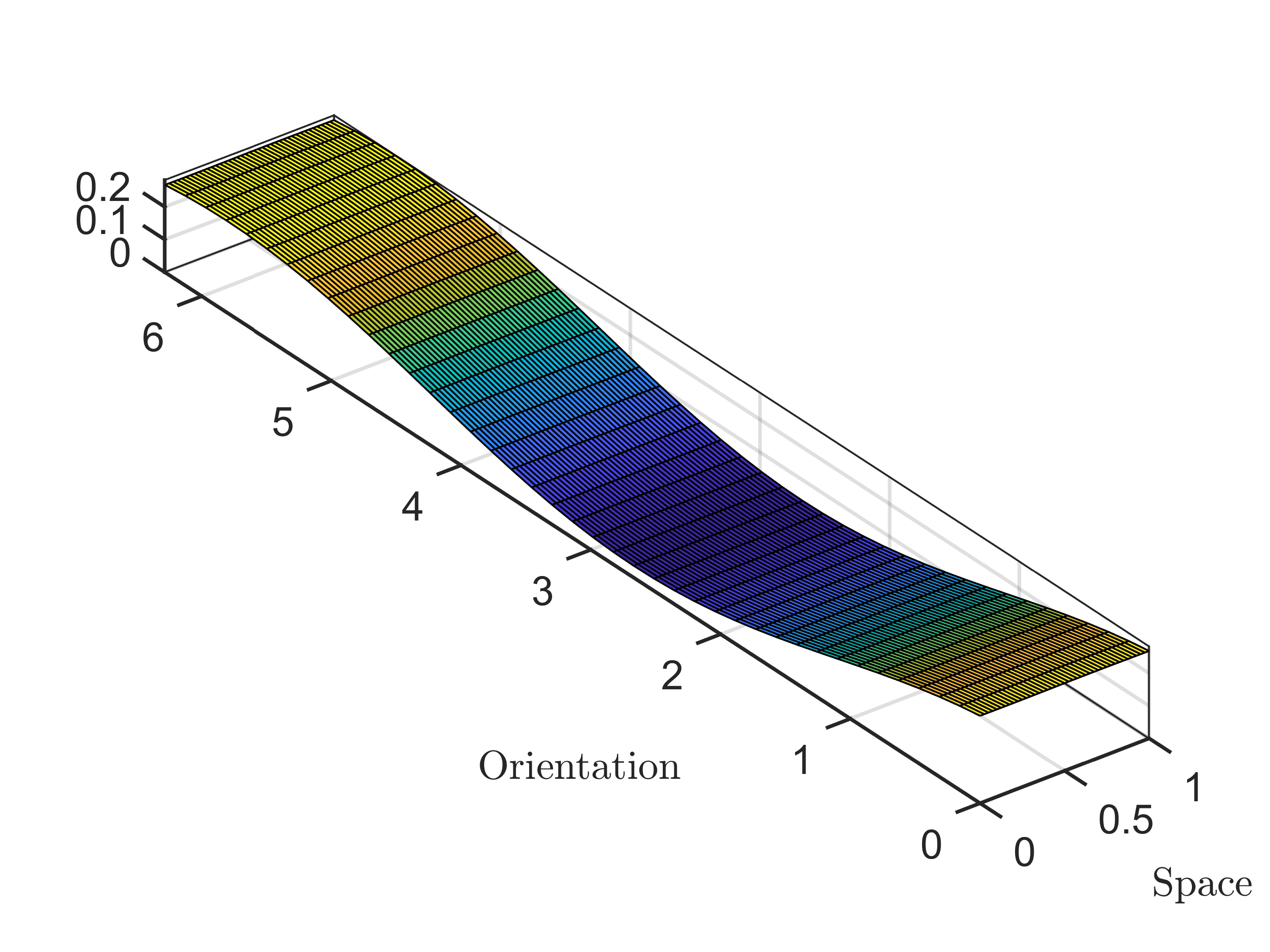}
    \caption{$\gamma = 0.38$}
  \end{subfigure}
  \caption{Steady-state density for the partially normalised (in $x$) interaction.
  The critical threshold is $\gamma_c = \Gamma/\pi \approx 0.32$.
  }
  \label{fig:x-norm}
\end{figure}

\section{Stability analysis for $h \neq 0$}\label{sec:h-neq-0}

When $h \neq 0$, the uniform density $\rho_0 = \frac{1}{2\pi}$ is no longer a stationary solution of the Fokker-Planck equation, because the confining field $-h\sin(\theta)$ generates an angular current even in the absence of interactions. The stability problem must therefore be formulated around the corresponding $h$-dependent stationary state.

In this section we first construct the stationary branch $\rho_h = \rho_0 + h\rho_1 + O(h^2)$ perturbatively for small $h$, then linearise around it and apply eigenvalue perturbation theory in the sense of Kato~\cite{kato1966perturbation} to obtain an explicit formula for the critical coupling $\gamma_c(h)$. Throughout, we work in the spatially homogeneous setting and restrict attention to the fully normalised monochromatic case ($\kappa = 1$, $n = 1$, $a_1 = 1$). The extension to general normalisations is discussed in Remark~\ref{rem:general-kappa}.

The spatially homogeneous stationary Fokker-Planck equation takes the form
\begin{equation}\label{eq:stat-FP-homog}
    0 = -\partial_\theta\Big(\big[F - h\sin(\theta)
    - \gamma\big(\sin(\theta)\, r_c - \cos(\theta)\, r_s\big)\big]\,\rho\Big)
    + \Gamma\,\partial_\theta^2\rho,
\end{equation}
where $r_c = \int_0^{2\pi}\cos(\theta)\,\rho\dd\theta$ and $r_s = \int_0^{2\pi}\sin(\theta)\,\rho\dd\theta$ are the first-order Fourier moments of $\rho$. For $h = 0$ this reduces to the equation studied in Section~\ref{sec:h0}; the goal here is to track how its stationary solutions and their stability properties deform as $h$ is increased from zero.

\subsection{Construction of the stationary branch}\label{subsec:rho-h}
We look for a stationary solution of \eqref{eq:stat-FP-homog} in the form
\begin{equation}\label{eq:rho-h-expansion}
    \rho_h(\theta) = \rho_0 + h\,\rho_1(\theta) + O(h^2), \qquad
    \rho_0 = \frac{1}{2\pi},
\end{equation}
where normalisation requires $\int_0^{2\pi}\rho_1(\theta)\dd\theta = 0$. Since the first Fourier modes of $\rho_0$ vanish, the moments of $\rho_h$ expand as
\[
    r_{c,h} = h\,r_{c,1} + O(h^2), \qquad r_{s,h} = h\,r_{s,1} + O(h^2),
\]
where $r_{c,1} = \int_0^{2\pi}\cos(\theta)\,\rho_1\dd\theta$ and $r_{s,1} = \int_0^{2\pi}\sin(\theta)\,\rho_1\dd\theta$.

Inserting \eqref{eq:rho-h-expansion} into \eqref{eq:stat-FP-homog} and collecting terms at order $O(h)$ gives
\begin{equation}\label{eq:rho1-ODE-moments}
    \Gamma\,\partial_\theta^2 \rho_1
    - F\,\partial_\theta \rho_1
    + \frac{\gamma}{2\pi}\Big(r_{c,1}\cos(\theta) + r_{s,1}\sin(\theta)\Big)
    + \frac{1}{2\pi}\cos(\theta) = 0.
\end{equation}
Here the term involving $r_{c,1}$ and $r_{s,1}$ arises from the linearisation of the interaction, and the forcing $\frac{1}{2\pi}\cos(\theta)$ is the $O(h)$ contribution of the confining field acting on $\rho_0$, via $-\partial_\theta(\frac{1}{2\pi}\cdot(-\sin(\theta))) = \frac{1}{2\pi}\cos(\theta)$.

Since \eqref{eq:rho1-ODE-moments} involves only first harmonics, we set $\rho_1 = A\cos(\theta) + B\sin(\theta)$, giving $r_{c,1} = \pi A$ and $r_{s,1} = \pi B$, so \eqref{eq:rho1-ODE-moments} becomes
\begin{equation}\label{eq:rho1-ODE}
    \Gamma\,\partial_\theta^2\rho_1
    - F\,\partial_\theta\rho_1
    + \frac{\gamma}{2}\rho_1
    + \frac{1}{2\pi}\cos(\theta) = 0.
\end{equation}
Substituting $\rho_1(\theta)$ into \eqref{eq:rho1-ODE} and matching coefficients of $\cos(\theta)$ and $\sin(\theta)$ gives the following unique values for $A$ and $B$
\begin{equation*}\label{eq:AB-coeffs}
    A = -\frac{1}{2\pi}\,\frac{\alpha}{\alpha^2 + F^2}, \qquad
    B = \frac{1}{2\pi}\,\frac{F}{\alpha^2 + F^2},
\end{equation*}
where $\alpha = \frac{\gamma}{2} - \Gamma$. Therefore
\begin{equation}\label{eq:rho1-explicit}
    \rho_1(\theta) = -\frac{1}{2\pi}\,\frac{\alpha}{\alpha^2 + F^2}\,\cos(\theta)
    + \frac{1}{2\pi}\,\frac{F}{\alpha^2 + F^2}\,\sin(\theta).
\end{equation}
For later use we record the complex Fourier representation $\rho_1(\theta) = c_1\,e^{i\theta} + c_{-1}\,e^{-i\theta}$, where
\begin{equation}\label{eq:c-pm1}
    c_{\pm 1} = \frac{A \mp iB}{2} = -\frac{1}{4\pi(\alpha \mp iF)}.
\end{equation}

The perturbative expansion \eqref{eq:rho-h-expansion} is valid as long as $h\|\rho_1\|$ remains small. For the stability analysis below we require $h^2\|\rho_1\|^2 \ll 1$, which holds provided $h^2 \ll \alpha^2 + F^2$. In particular, for any fixed $F > 0$ the expansion is valid for all sufficiently small $h$, uniformly in $\alpha$ near $0$.

\begin{rem}[General normalisation]\label{rem:general-kappa}
For a normalisation with spatial kernel constant $\kappa_0 = \int_{\R^2}\cKr(\by)\dd\by$, the same calculation applies with $\alpha$ replaced by $\alpha_0 := \frac{\gamma\kappa_0}{2} - \Gamma$. The formula \eqref{eq:rho1-explicit} holds with $\alpha_0$ in place of $\alpha$. We use $\alpha$ (i.e.\ $\kappa_0 = 1$) throughout this section for notational simplicity.
\end{rem}

\subsection{Linearisation around $\rho_h$}\label{subsec:linearisation-h}
We now linearise the Fokker-Planck equation around the stationary branch $\rho_h$ constructed above. The central question is how the critical coupling $\gamma_c$ shifts as $h$ increases from zero, and the answer is obtained by tracking the leading eigenvalue of the linearised operator as it deforms under the perturbation. The strategy is as follows: we decompose the linearised operator as $L_h = L_0 + hL_1 + O(h^2)$, where $L_0$ is the operator already diagonalised in Section~\ref{sec:h0}, and apply Kato's analytic perturbation theory~\cite{kato1966perturbation} to follow the critical eigenvalue of $L_0$ order by order in $h$. We will find that the first-order correction $\lambda_1$ vanishes identically — a structural cancellation rooted in the Fourier support of the perturbation — so the leading effect of the confining field on the threshold appears only at order $h^2$, through an explicit second-order correction $\lambda_2$ that depends on both $F$ and $\Gamma$.

Writing $\rho(\theta,t) = \rho_h(\theta) + \eta(\theta,t)$ with $\int_0^{2\pi}\eta\dd\theta = 0$, the evolution of the perturbation $\eta$ is governed by $\partial_t \eta = L_h \eta$, where
\begin{equation}\label{eq:Lh-operator}
    L_h\eta = -\partial_\theta\Big(\big[F - h\sin(\theta)
    - \gamma I[\rho_h]\big]\,\eta
    - \gamma\rho_h\, I[\eta]\Big) + \Gamma\,\partial_\theta^2\eta.
\end{equation}
Using the expansions $\rho_h = \rho_0 + h\rho_1 + O(h^2)$ and $I[\rho_h] = hI[\rho_1] + O(h^2)$, we decompose $L_h = L_0 + hL_1 + O(h^2)$, with
\begin{align}
    L_0\eta &= -F\,\partial_\theta\eta + \Gamma\,\partial_\theta^2\eta
    + \frac{\gamma}{2\pi}\,\partial_\theta I[\eta], \label{eq:L0-def}\\
    L_1\eta &= \partial_\theta\Big(
    \big[\sin(\theta) + \gamma I[\rho_1](\theta)\big]\,\eta
    + \gamma\rho_1(\theta)\, I[\eta]\Big). \label{eq:L1-def}
\end{align}
Here $L_0$ is the linearisation around $\rho_0$ at $h = 0$, already analysed in Section~\ref{sec:h0}, and $L_1$ captures the first-order correction induced by the confining field and the deformation of the stationary state.

The interaction operator $I[\cdot]$ acts on the complex Fourier basis $e_n(\theta) = e^{in\theta}$ as
\begin{equation*}\label{eq:I-on-en}
    I[e_1] = -i\pi\, e_1, \qquad
    I[e_{-1}] = i\pi\, e_{-1}, \qquad
    I[e_n] = 0 \quad (|n| \neq 1),
\end{equation*}
so $L_0$ is diagonal on the Fourier basis with eigenvalues
\begin{equation*}\label{eq:L0-spectrum}
    \mu_n = \begin{cases}
        \alpha - iF, & n = 1, \\
        \alpha + iF, & n = -1, \\
        -n^2\Gamma - inF, & |n| \geq 2.
    \end{cases}
\end{equation*}
The critical eigenvalues are $\mu_{\pm 1} = \alpha \mp iF$, whose common real part $\alpha$ changes sign at $\gamma = 2\Gamma$; all other eigenvalues satisfy $\mathrm{Re}\,\mu_n = -n^2\Gamma < 0$ and are therefore always stable.

We focus on the $n = +1$ branch; the $n = -1$ branch is obtained by complex conjugation. We apply the analytic perturbation theory of Kato~\cite{kato1966perturbation} to the simple eigenvalue $\lambda_0 = \mu_1 = \alpha - iF$ of $L_0$, expanding
\[
    \lambda(h) = \lambda_0 + h\lambda_1 + h^2\lambda_2 + O(h^3), \qquad
    \phi(h) = e_1 + h\phi_1 + h^2\phi_2 + O(h^3).
\]
Since $L_0$ is diagonal, its left and right eigenvectors coincide, and the first-order correction is given by
\[
    \lambda_1 = \frac{\langle e_1, L_1 e_1 \rangle}{\langle e_1, e_1 \rangle}.
\]

\begin{lem}\label{lem:lambda1-vanishes}
    $\lambda_1 = 0$.
\end{lem}

\begin{proof}
The operator $L_1$ defined in \eqref{eq:L1-def} consists of two terms, each of the form $\partial_\theta(\cdot)$, where the argument involves multiplying $e_1$ by first-harmonic functions: $\sin(\theta)$, $I[\rho_1](\theta)$, and $\rho_1(\theta)$. Each of these multipliers has Fourier support on modes $n = \pm 1$, so multiplying by $e_1$ produces only modes $e_0$ and $e_2$. The outer derivative $\partial_\theta$ annihilates the $e_0$ component, so $L_1 e_1 \in \mathrm{span}\{e_2\}$, which is orthogonal to $e_1$, giving $\lambda_1 = 0$.
\end{proof}
Since $\lambda_1 = 0$, the leading correction to the eigenvalue is at second order. By the Kato expansion~\cite{kato1966perturbation}, this is given by
\begin{equation}\label{eq:lambda2-general}
    \lambda_2 = \frac{\langle e_1, L_1\phi_1\rangle}{\langle e_1, e_1\rangle},
\end{equation}
where $\phi_1$ is determined by the first-order equation
\[
    (L_0 - \lambda_0)\phi_1 = -L_1 e_1.
\]

To evaluate \eqref{eq:lambda2-general}, we proceed in three steps. First, we compute $L_1e_1$. Since $\rho_1$, $I[\rho_1]$, and $\sin(\theta)$ contain only the Fourier modes $\pm1$, the operator $L_1$ maps $e_n$ into a combination of $e_{n-1}$ and $e_{n+1}$. In particular, $L_1e_1$ is proportional to $e_2$, so the first-order correction $\phi_1$ is also proportional to $e_2$. Second, we solve for $\phi_1$ using
\[
    (L_0-\lambda_0)\phi_1=-L_1e_1,
\]
which is immediate because $L_0$ is diagonal. Third, we substitute the resulting expression for $\phi_1$ into \eqref{eq:lambda2-general}; this shows that only the $e_1$ component of $L_1e_2$ contributes to $\lambda_2$. We therefore need only two scalar coefficients: the coefficient of $e_2$ in $L_1e_1$, and the coefficient of $e_1$ in $L_1e_2$.

Accordingly, define
\[
    a := \frac{\langle e_2, L_1 e_1\rangle}{\langle e_2, e_2\rangle}.
\]
Then $\phi_1 = -\frac{a}{\mu_2 - \lambda_0}\,e_2$, where $\mu_2 = -4\Gamma - 2iF$ is the eigenvalue of $L_0$ on $e_2$. Next define
\[
    b := \frac{\langle e_1, L_1 e_2\rangle}{\langle e_1, e_1\rangle}.
\]
Substituting the expression for $\phi_1$ into \eqref{eq:lambda2-general} gives
\begin{equation}\label{eq:lambda2-ab}
    \lambda_2 = -\frac{ab}{\mu_2 - \lambda_0}.
\end{equation}

It remains to compute $a$ and $b$ explicitly. Define $q^{(1)}(\theta) := -\sin(\theta) - \gamma I[\rho_1](\theta)$, so that $L_1\eta = -\partial_\theta(q^{(1)}\eta) + \gamma\partial_\theta(\rho_1\, I[\eta])$, and let $q_{\pm 1}$ denote the Fourier coefficients of $q^{(1)}$ on $e_{\pm 1}$. Applying $L_1$ to $e_1$ and using $I[e_1] = -i\pi\, e_1$ gives $L_1 e_1 = \big(-2iq_1 + 2\gamma\pi c_1\big)\,e_2$, so that $a = -2iq_1 + 2\gamma\pi c_1$. Substituting the explicit Fourier coefficients $q_1$ and $c_1$ from \eqref{eq:c-pm1} and using $\gamma = 2(\alpha + \Gamma)$, we obtain
\begin{equation*}\label{eq:coeff-a}
    a = \frac{F - i(2\Gamma + \alpha)}{F + i\alpha}.
\end{equation*}
Applying $L_1$ to $e_2$ and using $I[e_2] = 0$ gives $L_1 e_2 = -\partial_\theta(q^{(1)} e_2)$. Since $q^{(1)} e_2 = q_1\,e_3 + q_{-1}\,e_1$, differentiation gives the $e_1$-component $-iq_{-1}$, so $b = -iq_{-1}$, giving
\begin{equation}\label{eq:coeff-b}
    b = -\frac{F + i\Gamma}{2(F - i\alpha)}.
\end{equation}
Combining \eqref{eq:lambda2-ab}--\eqref{eq:coeff-b} with $\mu_2 - \lambda_0 = -(4\Gamma + \alpha + iF)$, we arrive at the second-order eigenvalue correction
\begin{equation}\label{eq:lambda2-formula}
    \lambda_2 = -\frac{ab}{\mu_2 - \lambda_0}
    = \frac{1}{4\Gamma + \alpha + iF}
    \left(-\frac{F - i(2\Gamma + \alpha)}{F + i\alpha}\right)
    \left(\frac{F + i\Gamma}{2(F - i\alpha)}\right).
\end{equation}
\subsection{The critical coupling $\gamma_c(h)$}\label{subsec:gamma-c}

The onset of instability is determined by the condition $\mathrm{Re}\,\lambda(h) = 0$. Since $\mathrm{Re}\,\lambda_0 = \alpha$ and $\lambda_1 = 0$, expanding to second order gives
\[
    \alpha + h^2\,\mathrm{Re}(\lambda_2) = 0.
\]
Writing $\gamma_c(h) = 2\Gamma + h^2\delta\gamma + O(h^4)$, so that $\alpha = h^2\delta\gamma/2 + O(h^4)$ at the critical point, and evaluating $\mathrm{Re}(\lambda_2)$ at the unperturbed threshold $\alpha = 0$, the criticality condition gives $\delta\gamma = -2\,\mathrm{Re}(\lambda_2)|_{\alpha=0}$.

At $\alpha = 0$, formula \eqref{eq:lambda2-formula} simplifies to
\[
    \lambda_2\big|_{\alpha=0}
    = \frac{1}{4\Gamma + iF}
    \cdot\left(-\frac{F - 2i\Gamma}{F}\right)
    \cdot\frac{F + i\Gamma}{2F}.
\]
Expanding the product of the last two factors gives numerator $-(F - 2i\Gamma)(F + i\Gamma) = -(F^2 + 2\Gamma^2) + iF\Gamma$ over denominator $2F^2$, so
\[
    \lambda_2\big|_{\alpha=0}
    = \frac{-(F^2 + 2\Gamma^2) + iF\Gamma}{2F^2(4\Gamma + iF)}.
\]
Multiplying through by $(4\Gamma - iF)$ and extracting the real part gives
\begin{equation}\label{eq:Re-lambda2}
    \mathrm{Re}(\lambda_2)\big|_{\alpha=0}
    = -\frac{\Gamma(3F^2 + 8\Gamma^2)}{2F^2(16\Gamma^2 + F^2)}.
\end{equation}
Since $\delta\gamma = -2\,\mathrm{Re}(\lambda_2)|_{\alpha=0}$, we arrive at the main result of this section.

\begin{prop}[Critical coupling with confinement]\label{thm:gamma-c}
In the spatially homogeneous regime with fully normalised interaction, the critical alignment strength satisfies
\begin{equation}\label{eq:gamma-c-formula}
    \gamma_c(h) = 2\Gamma
    + \frac{h^2\,\Gamma\,(3F^2 + 8\Gamma^2)}{F^2(16\Gamma^2 + F^2)}
    + O(h^4).
\end{equation}
\end{prop}

The correction is positive, so the confining field raises the instability threshold. Only even powers of $h$ appear, as expected from the symmetry $h \mapsto -h$. For large $F$ the correction is of order $F^{-2}$, so the effect of confinement weakens in the fast-tilt regime; conversely, the coefficient diverges as $F \to 0$, reflecting the breakdown of the perturbation expansion in that limit. For general normalisations, replacing $\alpha = \frac{\gamma}{2} - \Gamma$ by $\alpha_0 = \frac{\gamma\kappa_0}{2} - \Gamma$ shifts the unperturbed threshold to $\gamma_c(0) = \frac{2\Gamma}{\kappa_0}$ while leaving the structure of the $O(h^2)$ correction unchanged.

\subsection{Numerical verification}\label{subsec:numerics-h}

We verify the perturbative predictions using a Fourier-Galerkin discretisation of the linearised operator $L_h$ about the perturbative steady state $\rho_h^{\mathrm{pert}}$. Perturbations are expanded in the angular Fourier basis $\eta(\theta)=\sum_{\substack{n=-N\\ n\neq0}}^{N}\eta_n e^{in\theta}$,
and the expansion is truncated at $|n|\le N$ (we use $N=32$). In this basis the operator $L_h$ is represented by a complex matrix $M_h\in\mathbb{C}^{(2N)\times(2N)}$.

The matrix entries are obtained directly from the Fourier representation of the linearised operator
\[
L_h\eta
=
-\partial_\theta\!\big(q_h\,\eta\big)
+\gamma\,\partial_\theta\!\big(\rho_h^{\mathrm{pert}} I[\eta]\big)
+\Gamma\,\partial_\theta^2\eta ,
\]
where
\[
q_h(\theta)=F-h\sin(\theta)-\gamma\,I[\rho_h^{\mathrm{pert}}](\theta),
\qquad
\rho_h^{\mathrm{pert}}(\theta)=\frac{1}{2\pi}+h\,\rho_1(\theta).
\]
Because both $q_h$ and $\rho_h^{\mathrm{pert}}$ contain only the Fourier modes $0,\pm1$, multiplication by these functions couples each Fourier mode only to its nearest neighbours. This structure gives a sparse banded Galerkin matrix.

For fixed parameters $(\Gamma,F,\gamma)$ we compute the eigenvalues of $M_h$ and track the spectral branch associated with the $n=+1$ mode. At $h=0$ this branch is identified as the eigenvalue closest to
\[
\lambda_0=\alpha-iF,
\qquad
\alpha=\gamma/2-\Gamma,
\]
and as $h$ increases we continue it by selecting the eigenvalue closest to the previously identified one in the complex plane.

The resulting numerical eigenvalue $\lambda^{\mathrm{num}}(h)$ is compared with the perturbative prediction
\[
\lambda^{\mathrm{pert}}(h)
=
\lambda_0+h^2\lambda_2+O(h^3),
\]
where $\lambda_2$ is obtained from the second-order perturbation calculation described in Section~\ref{subsec:linearisation-h}. Figure~\ref{fig:critical-branch} shows excellent agreement between the numerical eigenvalue branch and the perturbative prediction for small $h$. In particular, the real part exhibits the quadratic dependence on $h$ predicted by the theory.
\begin{figure}[H]
    \centering
    \includegraphics[width=\textwidth]{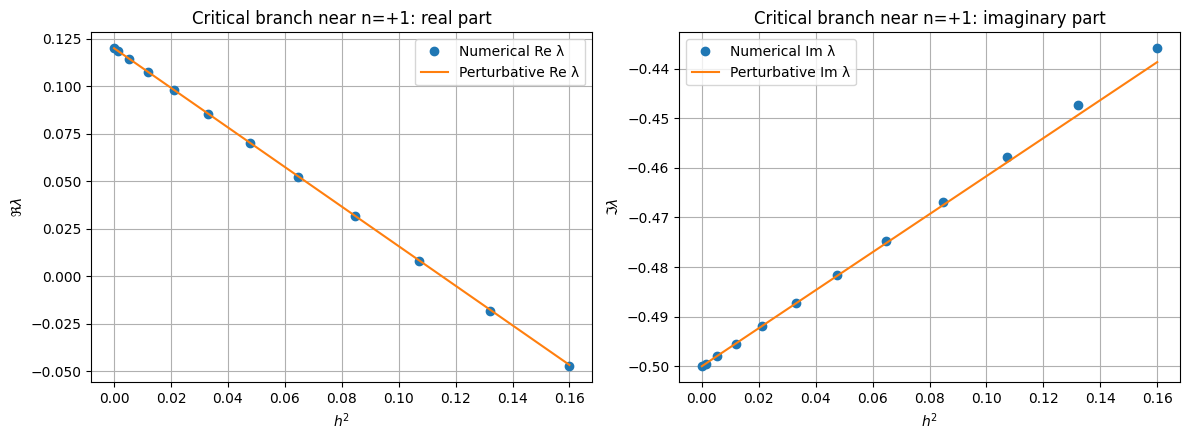}
    \caption{
    Real and imaginary parts of the eigenvalue branch associated with the $n=+1$ mode as a function of $h^2$ for $F=0.5$ and $\Gamma=1$. Markers denote numerical eigenvalues from the Fourier--Galerkin discretisation, while solid lines show the perturbative prediction $\lambda_+(h)=\lambda_0+h^2\lambda_2$.
}
    \label{fig:critical-branch}
\end{figure}

We also compute the critical coupling $\gamma_c(h)$ at which the real part of the $n=+1$ branch crosses zero. For each fixed value of $h$ we solve
$\mathrm{Re}\,\lambda_+(h,\gamma)=0$ numerically using a bisection method in $\gamma$, where $\lambda_+(h,\gamma)$ denotes the tracked continuation of the $n=+1$ eigenvalue. This gives the numerical threshold $\gamma_c^{\mathrm{num}}(h)$.

For comparison we also compute a perturbative prediction $\gamma_c^{\mathrm{pert}}(h)$ obtained by solving
\[
\mathrm{Re}\,\lambda^{\mathrm{pert}}_+(h,\gamma)=0,
\qquad
\lambda^{\mathrm{pert}}_+(h,\gamma)
=
\alpha-iF+h^2\lambda_2(\alpha),
\]
with $\alpha=\gamma/2-\Gamma$. This provides a perturbative approximation that incorporates the dependence of $\lambda_2$ on $\alpha$.

Figure~\ref{fig:gamma-c-h} compares the numerical and perturbative thresholds for $F=0.5$ and $\Gamma=1$, and shows excellent agreement for small $h$.
\begin{figure}[H]
    \centering
    \includegraphics[width=0.7\textwidth]{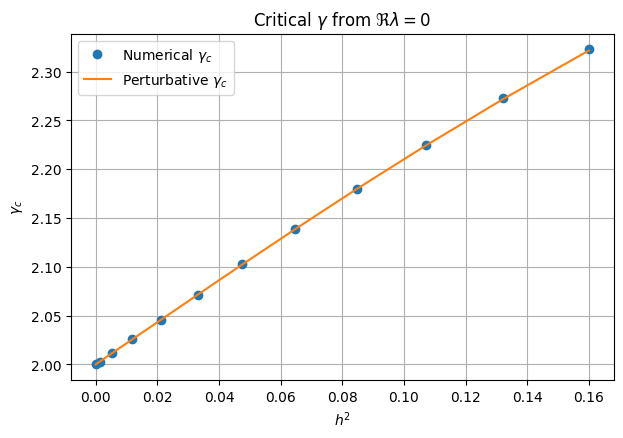}
    \caption{
    Critical coupling $\gamma_c(h)$ as a function of $h^2$ for $F=0.5$ and $\Gamma=1$. Circles show the numerical threshold obtained from $\mathrm{Re}\,\lambda_+(h,\gamma)=0$, while the solid curve shows the perturbative prediction based on $\lambda_+(h)=\lambda_0+h^2\lambda_2$.
    }
    \label{fig:gamma-c-h}
\end{figure}
\section{Spatial dependence for $h \neq 0$}\label{sec:spatial}

The analysis of Section~\ref{sec:h-neq-0} was carried out in the spatially homogeneous setting. We now reinstate the spatial variable $\bx \in \R^2$ and examine whether $\bk = 0$ remains the most unstable mode when $h \neq 0$.

The full linearised equation around $\rho_h(\theta)$ is $\partial_t\eta = L_h\eta$, with
\begin{equation*}\label{eq:Lh-spatial}
    L_h\eta = -v_0(\cos(\theta),\sin(\theta))\cdot\nabla_\bx\eta
    - \partial_\theta\Big(\big[F - h\sin(\theta) - \gamma I[\rho_h](\theta)\big]\,\eta
    - \gamma\rho_h(\theta)\, I[\eta](\bx,\theta)\Big)
    + \Gamma\,\partial_\theta^2\eta,
\end{equation*}
where $I[\eta](\bx,\theta) = \int_{\R^2}\int_0^{2\pi}\cKr(\bx - \bx')\sin(\theta - \theta')\eta(\bx',\theta')\dd\theta'\dd\bx'$. Since $\rho_h$ is spatially homogeneous, $I[\rho_h](\bx,\theta) = \kappa_0 I_0[\rho_h](\theta)$ with $\kappa_0 := \int_{\R^2}\cKr(\by)\dd\by$.

Decomposing in spatial Fourier modes, convolution becomes multiplication and the dynamics decouples across wavenumbers: $\partial_t\hat\eta(\bk,\cdot,t) = L_h(\bk)\hat\eta(\bk,\cdot,t)$, where
\begin{equation*}\label{eq:Lh-k}
    L_h(\bk)\phi = -iv_0(k_1\cos(\theta) + k_2\sin(\theta))\,\phi
    - \partial_\theta\Big(\big[F - h\sin(\theta) - \gamma\kappa_0 I_0[\rho_h](\theta)\big]\,\phi
    - \gamma\rho_h(\theta)\, I_\bk[\phi](\theta)\Big)
    + \Gamma\,\partial_\theta^2\phi,
\end{equation*}
with $I_\bk[\phi](\theta) := \widehat\cKr(\bk)\int_0^{2\pi}\sin(\theta-\theta')\phi(\theta')\dd\theta'$. At $\bk = 0$, the transport term vanishes and $\widehat\cKr(\mathbf{0}) = \kappa_0$, so $L_h(\mathbf{0})$ reduces to the homogeneous operator of Section~\ref{sec:h-neq-0}, and the results of proposition~\ref{thm:gamma-c} carry over directly.

For $\bk \neq 0$, two structural observations suggest that the growth rate $\mathrm{Re}\,\lambda(\bk,h)$ is strictly less than $\mathrm{Re}\,\lambda(\mathbf{0},h)$. First, the transport term is multiplication by a purely imaginary function and is therefore skew-adjoint on $L^2(0,2\pi)$, contributing nothing to growth rates at leading order. Second, since $\cKr \in L^1(\R^2)$ is non-negative and even, $|\widehat\cKr(\bk)| \leq \widehat\cKr(\mathbf{0}) = \kappa_0$ for all $\bk$, so the effective alignment interaction is strictly weaker at $\bk \neq 0$ than at $\bk = 0$. Both conditions are satisfied by the kernels considered in this paper; in particular, for $\cKr(\bx) = \mathds{1}_{|\bx|\leq R}$ one has $\widehat\cKr(\bk) = \pi R^2 S_R(|\bk|) < \kappa_0$ for all $\bk \neq 0$.

Together these observations indicate that $\bk = 0$ remains the most unstable mode and that $\gamma_c(h)$ determines the onset of instability for the full spatially inhomogeneous problem, though a rigorous proof would require resolvent estimates for $L_h(\bk)$ and we leave this as an open problem.

\section{Conclusion}\label{sec:conclusion}
In this paper we have analysed the onset of synchronisation in a spatially extended Kuramoto-type model of interacting particles subject to an external field of strength $h$, a tilt $F$, and noise of strength $\Gamma$. Our main contributions are as follows.

First, for $h=0$ we carried out a full PDE-level linear stability analysis of the disordered state across four natural normalisations of the spatial interaction kernel. In each case the dominant unstable mode is the spatially homogeneous Fourier mode $\bk = 0$, providing a rigorous justification for the spatially homogeneous reduction employed in earlier work, and confirming that the tilt $F$ plays no role at threshold when $h=0$.

Second, for small $h>0$ we constructed the bifurcating branch of stationary solutions perturbatively and applied eigenvalue perturbation theory to track the critical coupling. The first-order correction vanishes, and the second-order calculation gives the explicit expansion
\begin{equation}
    \gamma_c(h) = 2\Gamma + \frac{h^2 \Gamma(3F^2 + 8\Gamma^2)}{F^2(16\Gamma^2 + F^2)} + O(h^4),
\end{equation}
demonstrating that the external field increases the synchronization threshold quadratically in $h$. The coefficient depends explicitly on both $F$ and $\Gamma$, and the formula is in good agreement with direct Fourier-Galerkin numerical computations.

Third, returning to the full spatially dependent problem for $h \neq 0$, we identified an obstruction to a finite-dimensional Kato reduction for $\bk \neq 0$ modes, arising from the coupling of neighbouring angular modes by the self-propulsion term. A heuristic argument nevertheless indicates that $\bk = 0$ remains the most unstable mode, so that the threshold $\gamma_c(h)$ derived in the homogeneous setting continues to govern the onset of instability in the full problem. Making this argument rigorous remains an open problem. More broadly, the rigorous and systematic study of non-equilibrium phase transitions for models in active matter remains a largely unexplored area.

\paragraph{Acknowledgments} BB is funded by a studentship from the Imperial College London EPSRC DTP in Mathematical Sciences Grant No. EP/W523872/1. GAP is partially supported by an ERC-EPSRC Frontier Research
Guarantee through grant no. EP/X038645, ERC through Advanced grant no. 247031, and a Leverhulme Trust Senior Research Fellowship, SRF{$\setminus$}R1{$\setminus$}241055. BB is grateful to Eloise Lardet for useful discussions.

\printbibliography

@article{chepizhko2021revisiting,
  title={Revisiting the emergence of order in active matter},
  author={Chepizhko, Oleksandr and Saintillan, David and Peruani, Fernando},
  journal={Soft Matter},
  volume={17},
  number={11},
  pages={3113--3120},
  year={2021},
  publisher={Royal Society of Chemistry}
}

@article{ramaswamy2010mechanics,
  title={The mechanics and statistics of active matter},
  author={Ramaswamy, Sriram},
  journal={Annual Review of Condensed Matter Physics},
  volume={1},
  number={1},
  pages={323--345},
  year={2010},
  publisher={Annual Reviews}
}

@book {pavliotis2014book,
    AUTHOR = {Pavliotis, Grigorios A.},
     TITLE = {Stochastic processes and applications},
    SERIES = {Texts in Applied Mathematics},
    VOLUME = {60},
      NOTE = {Diffusion processes, the Fokker-Planck and Langevin equations},
 PUBLISHER = {Springer, New York},
      YEAR = {2014},
     PAGES = {xiv+339},
      ISBN = {978-1-4939-1322-0; 978-1-4939-1323-7},
   MRCLASS = {60-01 (35Q84 35R60 60H10 60H35 60J60 82C31)},
  MRNUMBER = {3288096},
MRREVIEWER = {Isamu\ D\^{o}ku},
       DOI = {10.1007/978-1-4939-1323-7},
       URL = {https://doi.org/10.1007/978-1-4939-1323-7},
}

@article{Reimann_al2008,
Author = {Evstigneev, M. and Zvyagolskaya, O. and Bleil, S. and
   Eichhorn, R. and Bechinger, C. and Reimann, P.},
Title = {{Diffusion of colloidal particles in a tilted periodic potential: Theory
   versus experiment}},
Journal = {Physical Review E},
Year = {2008},
Volume = {{77}},
Number = {{4, Part 1}},
Month = {04},
DOI = {{10.1103/PhysRevE.77.041107}},
Article-Number = {{041107}},
ISSN = {{1539-3755}},
Unique-ID = {{ISI:000255456900010}},
}

@article {MartzelAslangul2001,
    AUTHOR = {Martzel, N. and Aslangul, C.},
     TITLE = {Mean-field treatment of the many-body {F}okker-{P}lanck
              equation},
   JOURNAL = {Journal of Physics. A. Mathematical and General},
  FJOURNAL = {Journal of Physics. A. Mathematical and General},
    VOLUME = {34},
      YEAR = {2001},
    NUMBER = {50},
     PAGES = {11225--11240},
      ISSN = {0305-4470},
   MRCLASS = {82C31 (82C35)},
  MRNUMBER = {1874137},
MRREVIEWER = {Vitor R. Vieira},
       DOI = {10.1088/0305-4470/34/50/305},
       URL = {http://dx.doi.org/10.1088/0305-4470/34/50/305},
}

@article{chaintron2022propagationII,
  title={Propagation of chaos: {A} review of models, methods and applications. {II}. {A}pplications},
  author={Chaintron, Louis-Pierre and Diez, Antoine},
  journal={Kinetic and Related Models},
  volume={15},
  number={6},
  pages={1017--1173},
  year={2022}
}

@article {Barre_2020,
    AUTHOR = {Barr\'{e}, J. and Bernardin, C. and Ch\'{e}trite, R. and
              Chopra, Y. and Mariani, M.},
     TITLE = {From fluctuating kinetics to fluctuating hydrodynamics: a
              {$\Gamma$}-convergence of large deviations functionals
              approach},
   JOURNAL = {Journal of Statistical Physics},
  FJOURNAL = {Journal of Statistical Physics},
    VOLUME = {180},
      YEAR = {2020},
    NUMBER = {1-6},
     PAGES = {1095--1127},
      ISSN = {0022-4715,1572-9613},
   MRCLASS = {82C22 (35Q82 76R50)},
  MRNUMBER = {4131027},
       DOI = {10.1007/s10955-020-02598-w},
       URL = {https://doi.org/10.1007/s10955-020-02598-w},
}

@misc{leimkuhler2025clusterformationdiffusivesystems,
      title={Cluster Formation in Diffusive Systems}, 
      author={Benedict Leimkuhler and René Lohmann and Grigorios A. Pavliotis and Peter A. Whalley},
      year={2025},
      eprint={2510.25034},
      archivePrefix={arXiv},
      primaryClass={math.NA},
      url={https://arxiv.org/abs/2510.25034}, 
}

@misc{gerber2025formationclusterscoarseningweakly,
      title={Formation of clusters and coarsening in weakly interacting diffusions}, 
      author={Nicolai Gerber and Rishabh Gvalani and Martin Hairer and Greg Pavliotis and André Schlichting},
      year={2025},
      eprint={2510.17629},
      archivePrefix={arXiv},
      primaryClass={math.AP},
      url={https://arxiv.org/abs/2510.17629}, 
}

@article{chaintron2022propagationI,
  title={Propagation of chaos: {A} review of models, methods and applications. {I}. {M}odels and methods},
  author={Chaintron, Louis-Pierre and Diez, Antoine},
  journal={Kinetic and Related Models},
  volume={15},
  number={6},
  pages={895--1015},
  year={2022}
}

@article {reimann_al02,
    AUTHOR = {P. Reimann and C. Van den Broeck and H. Linke and P. H\"{a}nggi and J.M. Rubi and A.
    Perez-Madrid},
     TITLE = {Diffusion in tilted periodic potentials: enhancement, universality and scaling},
   JOURNAL = {Physical Review E},
    VOLUME = {65},
    NUMBER = {3},
      YEAR = {2002},
     PAGES = {031104},
}

@book {Straton63,
    AUTHOR = {R. L. Stratonovich},
     TITLE = {Topics in the Theory of Random Noise. {V}ol. {I}},
    SERIES = {Revised English edition. Translated from the Russian by
              Richard A. Silverman},
 PUBLISHER = {Gordon and Breach Science Publishers},
   ADDRESS = {New York},
      YEAR = {1963},
     PAGES = {xiii+292},
   MRCLASS = {60.90 (94.10)},
  MRNUMBER = {MR0158437 (28 \#1660)},
MRREVIEWER = {A. T. Bharucha-Reid},
}

@article{bertoli2024stability,
  title={Stability of stationary states for mean field models with multichromatic interaction potentials},
  author={Bertoli, Benedetta and Goddard, Benjamin D and Pavliotis, Grigorios A},
  journal={IMA Journal of Applied Mathematics},
  volume={89},
  number={5},
  pages={833--859},
  year={2024},
  publisher={Oxford University Press}
}

@article{vicsek2012collective,
  title={Collective motion},
  author={Vicsek, Tam{\'a}s and Zafeiris, Anna},
  journal={Physics Reports},
  volume={517},
  number={3-4},
  pages={71--140},
  year={2012},
  publisher={Elsevier}
}

@article{marchetti2013hydrodynamics,
  title={Hydrodynamics of soft active matter},
  author={Marchetti, M Cristina and Joanny, Jean-François and Ramaswamy, Sriram and Liverpool, Tanniemola B and Prost, Jacques and Rao, Madan and Simha, R Aditi},
  journal={Reviews of Modern Physics},
  volume={85},
  number={3},
  pages={1143--1189},
  year={2013},
  publisher={APS}
}

@article{peruani2008mean,
  title={A mean-field theory for self-propelled particles interacting by velocity alignment mechanisms},
  author={Peruani, Fernando and Deutsch, Andreas and B{\"a}r, Markus},
  journal={The European Physical Journal Special Topics},
  volume={157},
  number={1},
  pages={111--122},
  year={2008},
  publisher={Springer}
}

@article{buendia2021hybrid,
  title={Hybrid-type synchronization transitions: {W}here incipient oscillations, scale-free avalanches, and bistability live together},
  author={Buend{\'\i}a, Victor and Villegas, Pablo and Burioni, Raffaella and Mu{\~n}oz, Miguel A},
  journal={Physical Review Research},
  volume={3},
  number={2},
  pages={023224},
  year={2021},
  publisher={APS}
}

@article{cheng2015long,
  title={The long time behavior of {B}rownian motion in tilted periodic potentials},
  author={Cheng, Liang and Yip, Nung Kwan},
  journal={Physica D: Nonlinear Phenomena},
  volume={297},
  pages={1--32},
  year={2015},
  publisher={Elsevier}
}

@article{peruani2012collective,
  title={Collective motion and nonequilibrium cluster formation in colonies of gliding bacteria},
  author={Peruani, Fernando and Starru{\ss}, J{\"o}rn and Jakovljevic, Vladimir and S{\o}gaard-Andersen, Lotte and Deutsch, Andreas and B{\"a}r, Markus},
  journal={Physical Review letters},
  volume={108},
  number={9},
  pages={098102},
  year={2012},
  publisher={APS}
}

@article{bechinger2016active,
  title={Active particles in complex and crowded environments},
  author={Bechinger, Clemens and Di Leonardo, Roberto and L{\"o}wen, Hartmut and Reichhardt, Charles and Volpe, Giorgio and Volpe, Giovanni},
  journal={Reviews of Modern Physics},
  volume={88},
  number={4},
  pages={045006},
  year={2016},
  publisher={APS}
}

@article{degond2008continuum,
  title={Continuum limit of self-driven particles with orientation interaction},
  author={Degond, Pierre and Motsch, S{\'e}bastien},
  journal={Mathematical Models and Methods in Applied Sciences},
  volume={18},
  number={supp01},
  pages={1193--1215},
  year={2008},
  publisher={World Scientific}
}

@article{chepizhko2010relation,
  title={On the relation between {V}icsek and {K}uramoto models of spontaneous synchronization},
  author={Chepizhko, AA and Kulinskii, VL},
  journal={Physica A: Statistical Mechanics and its Applications},
  volume={389},
  number={23},
  pages={5347--5352},
  year={2010},
  publisher={Elsevier}
}

@article{degond2014hydrodynamics,
  title={Hydrodynamics of the {Kuramoto--Vicsek} model of rotating self-propelled particles},
  author={Degond, Pierre and Dimarco, Giacomo and Mac, Thi Bich Ngoc},
  journal={Mathematical Models and Methods in Applied Sciences},
  volume={24},
  number={02},
  pages={277--325},
  year={2014},
  publisher={World Scientific}
}

@article{degond2015phase,
  title={Phase transitions, hysteresis, and hyperbolicity for self-organized alignment dynamics},
  author={Degond, Pierre and Frouvelle, Amic and Liu, Jian-Guo},
  journal={Archive for Rational Mechanics and Analysis},
  volume={216},
  number={1},
  pages={63--115},
  year={2015},
  publisher={Springer}
}

@article{carrillo2020long,
  title={Long-time behaviour and phase transitions for the {McKean--Vlasov} equation on the torus},
  author={Carrillo, Jose Antonio and Gvalani, Rishabh S and Pavliotis, Grigorios A and Schlichting, Andre},
  journal={Archive for Rational Mechanics and Analysis},
  volume={235},
  number={1},
  pages={635--690},
  year={2020},
  publisher={Springer}
}

@article {delgadino2020diffusive,
    AUTHOR = {Delgadino, M. G. and Gvalani, R. S. and Pavliotis,
              G. A.},
     TITLE = {On the {D}iffusive-{M}ean {F}ield {L}imit for {W}eakly
              {I}nteracting {D}iffusions {E}xhibiting {P}hase {T}ransitions},
   JOURNAL = {Arch. Ration. Mech. Anal.},
  FJOURNAL = {Archive for Rational Mechanics and Analysis},
    VOLUME = {241},
      YEAR = {2021},
    NUMBER = {1},
     PAGES = {91--148},
      ISSN = {0003-9527},
   MRCLASS = {82C22 (60J60 82C31)},
  MRNUMBER = {4271956},
       DOI = {10.1007/s00205-021-01648-1},
       URL = {https://doi.org/10.1007/s00205-021-01648-1},
}

@book {kato1966perturbation,
    AUTHOR = {T. Kato},
     TITLE = {Perturbation theory for linear operators},
    SERIES = {Classics in Mathematics},
      NOTE = {Reprint of the 1980 edition},
 PUBLISHER = {Springer-Verlag},
   ADDRESS = {Berlin},
      YEAR = {1995},
     PAGES = {xxii+619},
      ISBN = {3-540-58661-X},
   MRCLASS = {47A55 (46-00 47-00)},
  MRNUMBER = {MR1335452 (96a:47025)},
}

@article{shinomoto1986phase,
  title={Phase transitions in active rotator systems},
  author={Shinomoto, Shigeru and Kuramoto, Yoshiki},
  journal={Progress of Theoretical Physics},
  volume={75},
  number={5},
  pages={1105--1110},
  year={1986},
  publisher={Oxford University Press}
}

@article{odor2023synchronization,
  title={Synchronization transitions on connectome graphs with external force},
  author={{\'O}dor, G{\'e}za and Papp, Istv{\'a}n and Deng, Shengfeng and Kelling, Jeffrey},
  journal={Frontiers in Physics},
  volume={11},
  pages={1150246},
  year={2023},
  publisher={Frontiers Media SA}
}

@article{vicsek1995novel,
  title={Novel type of phase transition in a system of self-driven particles},
  author={Vicsek, Tam{\'a}s and Czir{\'o}k, Andr{\'a}s and Ben-Jacob, Eshel and Cohen, Inon and Shochet, Ofer},
  journal={Physical Review Letters},
  volume={75},
  number={6},
  pages={1226},
  year={1995},
  publisher={APS}
}

@article{bertoli2025phase,
  title={Phase transitions for interacting particle systems on random graphs},
  author={Bertoli, Benedetta and Pavliotis, Grigorios A and Zagli, Niccol{\`o}},
  journal={arXiv preprint arXiv:2504.02721},
  year={2025}
}

@article{bertini2010dynamical,
  title={Dynamical aspects of mean field plane rotators and the {K}uramoto model},
  author={Bertini, Lorenzo and Giacomin, Giambattista and Pakdaman, Khashayar},
  journal={Journal of Statistical Physics},
  volume={138},
  number={1},
  pages={270--290},
  year={2010},
  publisher={Springer}
}

@article{ohta2008critical,
  title={Critical phenomena in globally coupled excitable elements},
  author={Ohta, Hiroki and Sasa, Shin-Ichi},
  journal={Physical Review E—Statistical, Nonlinear, and Soft Matter Physics},
  volume={78},
  number={6},
  pages={065101},
  year={2008},
  publisher={APS}
}

@article{chate2008collective,
  title={Collective motion of self-propelled particles interacting without cohesion},
  author={Chat{\'e}, Hugues and Ginelli, Francesco and Gr{\'e}goire, Guillaume and Raynaud, Franck},
  journal={Physical Review E—Statistical, Nonlinear, and Soft Matter Physics},
  volume={77},
  number={4},
  pages={046113},
  year={2008},
  publisher={APS}
}

@article{dawson1983critical,
  title={Critical dynamics and fluctuations for a mean-field model of cooperative behavior},
  author={Dawson, Donald A},
  journal={Journal of Statistical Physics},
  volume={31},
  number={1},
  pages={29--85},
  year={1983},
  publisher={Springer}
}

@article{gregoire2004onset,
  title={Onset of collective and cohesive motion},
  author={Gr{\'e}goire, Guillaume and Chat{\'e}, Hugues},
  journal={Physical Review Letters},
  volume={92},
  number={2},
  pages={025702},
  year={2004},
  publisher={APS}
}

\end{document}